\documentclass[twocolumn,pra,nofootinbib]{revtex4-1}

\usepackage{amsthm}
\usepackage{amsmath}
\usepackage{bbold}

\usepackage{tikz}
\usepackage{pgfplots}
\usepackage{pgf}
\usetikzlibrary{patterns}
\usepackage{float}
\usepackage{graphicx}

 \theoremstyle{plain}
  \newtheorem{thm}{Theorem}
  
  \newtheorem{lem}[thm]{Lemma}
  \newtheorem{cor}[thm]{Corollary}
  
  \newtheorem*{thm*}{Theorem}
  \newtheorem*{prop*}{Proposition}
  \newtheorem*{lem*}{Lemma}
  \newtheorem*{cor*}{Corollary}
  \newtheorem*{remark*}{Remark}

\theoremstyle{definition}
\newtheorem{defn}[thm]{Definition}
 \newtheorem*{conj*}{Conjecture}

\begin{document}

\title{de Finetti reductions for correlations}

\author{Rotem Arnon-Friedman}
\author{Renato Renner}
\affiliation {Institute for Theoretical Physics, ETH-Z\"urich, CH-8093, Z\"urich, Switzerland}

\begin{abstract}

When analysing quantum information processing protocols one has to deal with large entangled systems, each consisting of many subsystems. To make this analysis feasible, it is often necessary to identify some additional structure. de Finetti theorems provide such a structure for the case where certain symmetries hold. More precisely, they relate states that are invariant under permutations of subsystems to states in which the subsystems are independent of each other. This relation plays an important role in various areas, e.g., in quantum cryptography or state tomography, where permutation invariant systems are ubiquitous. 
The known de Finetti theorems usually refer to the internal quantum state of a system and depend on its dimension. Here we prove a different de Finetti theorem where systems are modelled in terms of their statistics under measurements. This is necessary for a large class of applications widely considered today, such as device independent protocols, where the underlying systems and the dimensions are unknown and the entire analysis is based on the observed correlations.

\end{abstract}

\maketitle

\section{Introduction}

The analysis of quantum information processing protocols is a challenging task. Let it be a quantum tomography process, transmission of quantum information over a noisy channel or a cryptographic protocol -- all need to be analysed under general conditions. Since one usually has limited information about the actual quantum state given as input, the analysis should be valid for any given quantum state. For example, a cryptographic protocol should be proven secure independently of the input state, which can be chosen by a malicious adversary. As the space of all possible states can be very large and the structure of the states therein might be complicated due to entanglement, this task can be tedious in the good case, and infeasible in the worst. 

The quantum de Finetti theorems \cite{hudson1976locally, raggio1989quantum, caves2002unknown, renner2007symmetry} and the post selection theorem \cite{christandl2009postselection} address the above problem, by exploiting the symmetry of the considered states, namely permutation invariance. These mathematical tools allow us to simplify the analysis of quantum information processing tasks by reducing permutation invariant quantum states to a more structured state, called the quantum de Finetti state.
In general, we say that a state is of de Finetti-type if it is a convex combination of independent identically distributed (i.i.d.\@) states. 

de Finetti states are usually much easier to handle than general states due to their simple structure. Moreover, most established information-theoretic techniques can be applied only to i.i.d\ states, and therefore while not applicable when considering a general state, they can be used when considering de Finetti states. Therefore, a reduction to such states can simplify calculations and proofs of various quantum information processing tasks. Indeed, one of the famous applications of reductions to de Finetti states is a proof which states that in order to establish security of quantum key distribution against general attacks it is sufficient to consider attacks on individual signals~\cite{christandl2009postselection}. Other applications include quantum tomography \cite{christandl2012reliable} or quantum reverse Shannon coding \cite{berta2011reverse}. 

Unfortunately, the known variants of the quantum de Finetti theorems are not always applicable. A big class of protocols, commonly used in the past several years, to which those theorems are not applicable is the class of protocols in which the dimension of the states is unknown or cannot be bounded, and in particular, the class of device independent protocols (for a review on the topic, see for example \cite{scarani2012device, brunner2013bell}). The above mentioned theorems cannot be used in such cases for they depend on the dimension of the quantum state. 

In device independent cryptography \cite{mayers1998quantum,pironio2009device}, for example, one considers the devices as black boxes, about which we know nothing. The security of such protocols can therefore rely only on the observed statistics and not on the specific quantum states and measurements used in the protocol (in some protocols one does not even assume that the underlying physical system is restricted to be quantum! \cite{barrett2005no,hanggi2009quantum}). In these cases, one possible framework to work with is the framework of conditional probability distributions.  

Conditional probability distributions describe the operational behaviour of physical systems under measurements. That is, if we are only interested in modelling the measurement-outcome behaviour of our physical system, then the system can be described by a conditional probability distribution $\mathrm{P}_{A|X}$ where $X$ is the input, or the measurement performed on the system, and $A$ is the output. $\mathrm{P}_{A|X}(a|x)$ is the probability for outcome $a$ given that a measurement $x$ was made. We then say that $\mathrm{P}_{A|X}$ is the state of the system. Note that the state may have as many inputs and outputs as required and therefore we do not restrict the structure of the underlying system by describing it as a conditional probability distribution.

In quantum physics, for example, $\mathrm{P}_{A|X}$ is given by Born's rule. However, conditional probability distributions can also be used to describe states that might not conform with the theory of quantum physics, such as non-signalling states. Consider for example a state $\mathrm{P}_{AB|XY}$ shared by two space-like separated parties, Alice and Bob, each holding a subsystem of the state. $X$ and $A$ are then, respectively, the input and output of Alice, and $Y$ and $B$ of Bob. We then say that the state is non-signalling if it cannot be used to communicate, i.e., the output of one party is independent of the input of the other. The PR-box \cite{PR-box} is an example for a non-quantum bipartite state which can be written as a (non-signalling) conditional probability distribution. 

Given all the above, it is thus necessary to see whether de Finetti theorems are unique for quantum states or can be also proven on the level of the correlations in the framework of conditional probability distributions. More specifically, we are interested in a theorem that will allow us to reduce permutation invariant conditional probability distributions to a simple de Finetti-type conditional probability distribution, in a way that will be applicable in device independent protocols and, more generally, when the dimension of the underlying quantum states is unknown. Several different non-signalling de Finetti theorems have been established recently \cite{barrett2009finetti,christandl2009finite,brandao2012quantum}, but it is yet unknown how these can be applied to device independent cryptography\footnote{In most of these variants of de Finetti theorems, for example, it is assumed that the subsystems cannot signal each other. For current applications this is a too restrictive condition, since it is equivallent to assuming there is no memory in the devices.}. 

In this letter we prove a general de Finetti reduction theorem, from which we can derive several more specialised statements that are of interest for applications. The different reductions differ from one another in two main aspects -- the set of states to which they can be applied and the specific structure of the de Finetti state. Different de Finetti reductions can therefore be useful in different scenarios and under different assumptions. 

The simplest and most straightforward variant is a de Finetti reduction which can be applied to any permutation invariant conditional probability distribution. The second variant is a reduction which can be applied to a family of states which is relevant for cryptographic protocols based on the CHSH inequality \cite{CHSH} or the chained Bell inequalities \cite{braunstein1990wringing,Barrett2006Maximally}. There we connect any state $\mathrm{P}_{AB|XY}$ out of this family of states to a special \emph{non-signalling} de Finetti state $\tau^\mathcal{CHSH}_{AB|XY}$. We do not assume any non-signalling conditions between the subsystems of $\mathrm{P}_{AB|XY}$ and therefore the use of the de Finetti reduction is not restricted only to scenarios where each of the subsystems cannot signal each other. 

Up to date, almost all known device independent cryptographic protocols are based on the CHSH inequality or the more general chained Bell inequalities. For this reason we pay specific attention to states which are relevant for such protocols. However, our theorem can be applied also to other families of states which might be useful in future protocols. As an example of an application of our theorem we prove that for protocols which are based on the violation of the CHSH and chained Bell inequalities it is sufficient to consider the case where Alice and Bob share the de Finetti state $\tau^\mathcal{CHSH}_{AB|XY}$. We do this by bounding the distance between two channels which act on conditional probability distributions. 

In the following we start by describing and explaining the different de Finetti reductions. We then illustrate how the reductions can be used in applications. All the proofs are given in the Appendix.

\section{Results}

For stating the different de Finetti reductions we will need some basic definitions. $A$ and $X$ denote discrete random variables over $a \in \{0,1, ... , l-1\}^n$ and $x \in \{0,1, ... , m-1\}^n$ respectively. We use $[n]$ to denote the set $\{1,\dotsc,n\}$. An $n$-partite state $\mathrm{P}_{A|X}$ is a conditional probability distribution if for every $x$, $\sum_a \mathrm{P}_{A|X}(a|x)=1$ and for every $a, x$, $\mathrm{P}_{A|X}(a|x)\geq 0$. When we consider two different states $\mathrm{P}_{A|X}$ and $\mathrm{Q}_{A|X}$ it is understood that both states are over the same random variables $X$ and $A$. The de Finetti reductions deal with permutation invariant states and de Finetti states. Formally we define these as follows.
\begin{defn}\label{def:permutation}
	Given a state $\mathrm{P}_{A|X}$ and a permutation $\pi$ of its subsystems\footnote{Since we permute $a$ and $x$ together this is exactly as permuting the subsystems.} we denote by $\mathrm{P}_{A|X}\circ\pi$ the state which is defined by 
	\[
		\forall a,x \quad \left(\mathrm{P}_{A|X}\circ\pi \right) (a|x)=\mathrm{P}_{A|X}(\pi(a)|\pi(x)) \;.
	\]
	An $n$-partite state $\mathrm{P}_{A|X}$ is permutation invariant if for any permutation $\pi$, $\mathrm{P}_{A|X} = \mathrm{P}_{A|X}\circ\pi$.  
\end{defn}
As mentioned above, we say that a state is a de Finetti state if it is a convex combination of i.i.d.\ states. Formally, 
\begin{defn}
	A de Finetti state is a state of the form
	\[
		\tau_{A|X} = \int Q_{A_1|X_1}^{\otimes n} \mathrm{d}Q_{A_1|X_1}
	\]
	where $x_1\in \{0,1, ... , m-1\}$, $a_1 \in \{0,1, ... , l-1\}$, $\mathrm{d}Q_{A_1|X_1}$ is some measure on the space of 1-party states and $Q_{A_1|X_1}^{\otimes n}$ is a product of $n$ identical 1-party states $Q_{A_1|X_1}$, i.e., it is defined according to 
	\[
		Q_{A_1|X_1}^{\otimes n}(a|x) = \prod_{i\in[n]} Q_{A_1|X_1}(a_i|x_i) \;.
	\]
\end{defn}
As seen from the above definition, by choosing different measures $\mathrm{d}Q_{A_1|X_1}$ we define different de Finetti states. 

We are now ready to state the de Finetti reductions. For simplicity we start by giving the first corollary of the more general theorem (Theorem \ref{thm:post-selection}). This corollary is a reduction for conditional probability distributions, which connects general permutation invariant states to a specific de Finetti state. 

\begin{cor}[de Finetti reduction for conditional probability distributions]\label{cor:conditional}
	There exists a de Finetti state $\tau_{A|X}$ where $x \in \{ 0,1, ... ,m-1 \}^n$ and $a \in \{ 0,1, ... ,l-1 \} ^n$ such that for every permutation invariant state $\mathrm{P}_{A|X}$  
	\[
		\forall a,x \quad \mathrm{P}_{A|X} (a|x) \leq (n+1)^{m(l-1)} \; \tau_{A|X} (a|x)  \;.
	\]
\end{cor}

The de Finetti state $\tau_{A|X}$ is an \emph{explicit} state that we construct in the proof of the general theorem in Appendix~\ref{sec:general-proof}. The proof uses mainly combinatoric arguments; we choose $\tau_{A|X}$ in a specific way, such that a lower bound on $\tau_{A|X}(a|x)$ for all $a,x$ can be proven. We then use the permutation invariance of $\mathrm{P}_{A|X}$  to prove an upper bound on $\mathrm{P}_{A|X}(a|x)$. The result is then derived by combining the two bounds.

Corollary \ref{cor:conditional} holds for every permutation invariant state $\mathrm{P}_{A|X}$, not necessarily quantum or non-signalling. At first sight, the generality of the above mathematical statement might seem as a drawback in applications where only a restricted set of correlations is considered (e.g., only non-signalling correlations). Nevertheless, in a following work \cite{arnon2014nonsignalling} we show that this is not the case and apply this general theorem to prove parallel repetition theorems for non-signalling games. 
Note that according to Definition~\ref{def:permutation} we consider permutations which permute the 1-party subsystems of $\mathrm{P}_{A|X}$\footnote{This is in contrast to states $\mathrm{P}_{AB|XY}$ which can also be permuted as $\left(\mathrm{P}_{AB|XY}\circ\pi \right) (ab|xy)=\mathrm{P}_{AB|XY}\left(\pi(a)\pi(b)|\pi(x)\pi(y)\right)$, as is usually the case in cryptographic tasks. For dealing with such states we will consider a different reduction, stated as Corollary~\ref{cor:chsh-post-selection}.}. 

The multiplicative pre-factor of the de Finetti reduction, $(n+1)^{m(l-1)}$ in Corollary \ref{cor:conditional} for example, is relevant for applications. Intuitively, this is the ``cost'' for using $\tau_{A|X}$ instead of $\mathrm{P}_{A|X}$ in the analysis of the considered protocol. We therefore want it to be as small as possible. Nevertheless, as will be explained later, in many cases a pre-factor polynomial in $n$ suffices.  

Corollary \ref{cor:conditional} is relevant for scenarios in which one considers permutation invariant conditional probability distributions $\mathrm{P}_{A|X}$. However, if the states one considers have additional symmetries $\mathcal{S}$ then we can prove a better de Finetti reduction --- a reduction with a smaller pre-factor and a special de Finetti state with the same symmetries~$\mathcal{S}$.

In the following we consider a specific family of symmetries --- symmetries between different inputs and outputs of the subsystems of $\mathrm{P}_{A|X}$. Formally, the types of symmetries that we consider are described, among other things, by a number $d\leq m(l-1)$ which we call the degrees of freedom of the symmetry (see Appendix~\ref{sec:general-proof} for details and formal definition of the symmetries). More symmetry implies less degrees of freedom, i.e., smaller $d$, and as shown in the following theorem, this leads to a smaller pre-factor in the reduction. The general theorem then reads:

\begin{thm}[de Finetti reduction for conditional probability distributions with symmetries]\label{thm:post-selection}
	There exists a de Finetti state $\tau^\mathcal{S}_{A|X}$ where $x \in \{ 0,1, ... ,m-1 \}^n$ and $a \in \{ 0,1, ... ,l-1 \} ^n$ such that for every permutation invariant state $\mathrm{P}_{A|X}$ with symmetry $\mathcal{S}$ (with $d$ degrees of freedom) 
	\[
		\forall a,x \quad \mathrm{P}_{A|X} (a|x) \leq (n+1)^d \; \tau^\mathcal{S}_{A|X} (a|x) \;. 
	\]
\end{thm}
For the case of no symmetry we have $d=m(l-1)$ from which Corollary  \ref{cor:conditional} stated before follows. 

The symmetries $\mathcal{S}$ that we consider are of particular interest when considering cryptographic protocols which are based on non-signalling states. For example, the states which are relevant for protocols which are based on the violation of the CHSH inequality (such as \cite{masanes2009universally,hanggi2009quantum}) have a great amount of symmetry. The additional symmetry allows us to prove a corollary of Theorem \ref{thm:post-selection} which can be used to simplify such protocols. 

Before we state the corollary for the CHSH case, let us define what we mean when we say that a state has a CHSH-type symmetry. In cryptographic protocols based on the CHSH inequality the basic states that we consider are bipartite states $\mathrm{P}_{AB|XY}$ held by Alice and Bob where $a,b,x,y\in \{0,1\}^n$. 
\begin{defn}[CHSH-type symmetry]\label{def:chsh-symmetry}
	A state $\mathrm{P}_{AB|XY}$ where $a,b,x,y\in \{0,1\}^n$ has a CHSH-type symmetry if there exist $p_1,\dotsc,p_n\in [0,\frac{1}{2}]$ such that $\forall i\in\{1,\dotsc,n\}$,
	\begin{equation*}
	\begin{split}
		 & \forall a_i, b_i, x_i, y_i \\
		 & a_i\oplus\ b_i=x_i\cdot y_i \rightarrow \mathrm{P}_{AB|XY}(a_{\overline{i}}a_ib_{\overline{i}}b_i|x_{\overline{i}}x_iy_{\overline{i}}y_i) = \frac{1}{2}-p_i \\
		& a_i\oplus\ b_i \neq x_i\cdot y_i \rightarrow \mathrm{P}_{AB|XY}(a_{\overline{i}}a_ib_{\overline{i}}b_i|x_{\overline{i}}x_iy_{\overline{i}}y_i) = p_i \;.
	\end{split}
	\end{equation*}
	where $a_{\overline{i}}=a_1a_2\dotsc a_{i-1}a_{i+1} \dotsc a_n$ and $b_{\overline{i}}, x_{\overline{i}}, y_{\overline{i}}$ are defined in a similar way.
\end{defn}
A simple state $\mathrm{P}_{AB|XY}$ which has this symmetry for example is a product state of 2-partite states as in Figure~\ref{fig:CHSH_symmetry} with different values of $p$.

\begin{figure}
\begin{centering}
	\begin{tikzpicture}[scale=0.5]

		\draw[step=2] (-5,-4) grid (4,5);
		\draw[ultra thick] (-6,4)--(4,4);
		\draw[ultra thick] (-6,-4)--(4,-4);
		\draw[ultra thick] (-4,-4)--(-4,6);
		\draw[ultra thick] (4,-4)--(4,6);
		\draw[ultra thick] (-6,0)--(4,0);
		\draw[ultra thick] (0,-4)--(0,6);
		\draw (-4,4)--(-6,6);

		\draw (-3,3) node {$\frac{1}{2}-p$};
		\draw[red] (-1,3) node {$p$};
		\draw (1,3) node {$\frac{1}{2}-p$};
		\draw[red] (3,3) node {$p$};

		\draw[red] (-3,1) node {$p$};
		\draw (-1,1) node {$\frac{1}{2}-p$};
		\draw[red] (1,1) node {$p$};
		\draw (3,1) node {$\frac{1}{2}-p$};

		\draw (-3,-1) node {$\frac{1}{2}-p$};
		\draw[red] (-1,-1) node {$p$};
		\draw[red] (1,-1) node {$p$};
		\draw (3,-1) node {$\frac{1}{2}-p$};

		\draw[red] (-3,-3) node {$p$};
		\draw (-1,-3) node {$\frac{1}{2}-p$};
		\draw (1,-3) node {$\frac{1}{2}-p$};
		\draw[red] (3,-3) node {$p$};

		\draw (-5,3) node {0};
		\draw (-5,1) node {1};
		\draw (-5,-1) node {0};
		\draw (-5,-3) node {1};
		\draw (-6,2) node {0};
		\draw (-6,-2) node {1};

		\draw (-3,5) node {0};
		\draw (-1,5) node {1};
		\draw (1,5) node {0};
		\draw (3,5) node {1};
		\draw (-2,6) node {0};
		\draw (2,6) node {1};

		\draw (-5,4.4) node {$B_1$};
		\draw (-4.4,5) node {$A_1$};
		\draw (-6,5) node {$Y_1$};
		\draw (-5,6) node {$X_1$};
			
	\end{tikzpicture}
\end{centering}
\caption{A simple 2-partite state $P_{A_1B_1|X_1Y_1}$ with the CHSH symmetry.} \label{fig:CHSH_symmetry}
\end{figure} 

\begin{cor}[de Finetti reduction for states with the CHSH symmetry] \label{cor:chsh-post-selection}
	There exists a non-signalling de Finetti state $\tau^\mathcal{CHSH}_{AB|XY}$ where $a,b,x,y \in \{ 0,1 \}^n$ such that for every permutation invariant\footnote{Here a permutation acts on the bipartite state as $\left(\mathrm{P}_{AB|XY}\circ\pi \right) (ab|xy)=\mathrm{P}_{AB|XY}\left(\pi(a)\pi(b)|\pi(x)\pi(y)\right)$.} state $\mathrm{P}_{AB|XY}$ with the CHSH symmetry, for all $a,b,x,y$,
	\[
		\mathrm{P}_{AB|XY} (a,b|x,y) \leq (n+1) \; \tau^\mathcal{CHSH}_{AB|XY} (a,b|x,y)\;. 
	\]
\end{cor}
Note that we do not assume that the state $\mathrm{P}_{AB|XY}$ satisfies any non-signalling conditions. Our theorem holds even when there is signalling between the subsystems, and therefore can be used in a broad set of applications.

Corollary \ref{cor:chsh-post-selection} is derived from Theorem \ref{thm:post-selection} by showing that $d=1$ for the CHSH symmetry\footnote{Intuitivly, in the CHSH symmetry there is only one degree of freedom, i.e. $d=1$, since we are only free to choose one value $p$ when defining the basic CHSH state given in Figure \ref{fig:CHSH_symmetry}. Less symmetry implies more degrees of freedom. }. For pedagogical reasons, we also present a self-contained proof including an explicit construction of the state $\tau^\mathcal{CHSH}_{AB|XY}$ in Appendix~\ref{sec:chsh-proof}. 

Although the assumption about the symmetry of the states in Corollary \ref{cor:chsh-post-selection} appears to be rather restrictive, the statement turns out to be useful for applications. 

\section{Applications}

To illustrate the use of the de Finetti reductions, we start by considering the following abstract application. Let $\mathcal{T}$ be a test which interacts with a state $\mathrm{P}_{A|X}$ and outputs ``success'' or ``fail'' with some probabilities. One can think about this test, which can be chosen according to the application being considered, as a way to quantify the success probability of the protocol when the state $\mathrm{P}_{A|X}$ is given as input. For example, if one considers an estimation, or a tomography, protocol a test can be chosen to output ``success'' when the estimated state is close to the actual state \cite{christandl2009postselection}. 

We denote by $\mathrm{Pr}_{\text{fail}}(\mathrm{P}_{A|X})$ the probability that $\mathcal{T}$ outputs ``fail'' after interacting with $\mathrm{P}_{A|X}$. We consider permutation invariant tests, defined as follows. 
\begin{defn}\label{def:permutation-invariant-test}
	A test $\mathcal{T}$ is permutation invariant if for all states $\mathrm{P}_{A|X}$ and all permutations $\pi$ we have
	\[
	 	\mathrm{Pr}_{\text{fail}}(\mathrm{P}_{A|X}) = \mathrm{Pr}_{\text{fail}}(\mathrm{P}_{A|X}\circ\pi) \;.
	\]
\end{defn}

Using the de Finetti reduction in Corollary \ref{cor:conditional} we can prove upper bounds of the following type: 

\begin{lem}\label{lem:test-bound}
	Let $\mathcal{T}$ be a permutation invariant test. Then for every state $\mathrm{P}_{A|X}$ 
	\[
			\mathrm{Pr}_{\mathrm{fail}}(\mathrm{P}_{A|X}) \leq (n+1)^{m(l-1)} \mathrm{Pr}_{\mathrm{fail}}(\tau_{A|X}) \;.
	\]
\end{lem}

The importance of the de Finetti reductions is obvious from this abstract example --- if one wishes to prove an upper bound on the failure probability of the test $\mathcal{T}$, instead of proving it for all states $\mathrm{P}_{A|X}$ it is sufficient to prove it for the de Finetti state $\tau_{A|X}$ and ``pay'' for it with the additional polynomial pre-factor of $(n+1)^{m(l-1)}$. Since the de Finetti state has an i.i.d.\ structure this can highly simplify the calculations of the bound. 

Moreover, in many cases one finds that the bound on $\mathrm{Pr}_{\text{fail}}(\tau_{A|X})$ is exponentially small in $n$. For an estimation protocol, the failure probability of the test, when interacting with an i.i.d.\ state, can be shown to be exponentially small in the number of subsystems used for the estimation, using Chernoff bounds. This is also the case when dealing with security proofs -- the failure probability of a protocol, when a de Finetti state is given as input, is usually exponentially small in the number of subsystems used in the protocol. If this is indeed the case then the polynomial pre-factor of  $(n+1)^{m(l-1)}$ will not affect the bound in the asymptotic limit of large $n$. That is, an exponentially small bound on $\mathrm{Pr}_{\text{fail}}(\tau_{A|X})$ implies an exponentially small bound on $\mathrm{Pr}_{\text{fail}}(\mathrm{P}_{A|X})$.

For an estimation protocol as mentioned above the notion of the test, combined with the de Finetti reductions, can be used to prove that an estimation procedure of permutation invariant states succeeds with high probability. 

For readers who are interested in cryptography, we show in Appendix~\ref{sec:diamond-proofs} how to derive a similar result when considering the diamond norm~\cite{kitaev1997quantum}, i.e., the distance between channels acting on conditional probability distributions, instead of the abstract test $\mathcal{T}$. The diamond norm is the relevant distance measure when considering cryptographic protocols, and therefore using de Finetti reductions to upper bound the diamond norm can simplify the analysis of device independent  protocols. 

\section{Concluding remarks}
In this letter we introduced a general de Finetti-type theorem from which various more specialised variants can be derived. Crucially, such theorems can be formulated even without relying on assumptions regarding the non-signalling conditions between the subsystems or the underlying dimension. In the general theorem, Theorem~\ref{thm:post-selection}, we can also see how additional symmetries of the states can affect the pre-factor in the de Finetti reduction. This suggests that the same relationship might also exist in the quantum post selection theorem \cite{christandl2009postselection}, which is the quantum variant of the de Finetti reductions presented here.

As an example for an application we showed how our theorems can be used to bound the failure probability of a test. In a following work \cite{arnon2014nonsignalling} we show how to use the concept of the test, together with the de Finetti reduction given in Corollary \ref{cor:conditional} to prove parallel repetition results for non-local games. Previous de Finetti theorems could have not been used in the setting of non-local games due to their dependency on the dimension of the systems or the strict non-signalling conditions they assume. The new de Finetti theorem presented here therefore opens new possibilities and therefore strictly extends the range of applications to which de Finetti reductions can be applied.

As an additional example, we explain how our theorem can be used in device independent protocols in which the parties are not assumed to be restricted by quantum theory in Appendix~\ref{sec:diamond-proofs}. We hope that this approach will also be useful for quantum device independent information processing protocols in the future. One possible direction can be to use a similar de Finetti reduction as in Corollary \ref{cor:chsh-post-selection}, but for a Bell inequality in which the maximal violation is achieved within quantum theory. This way, the resulting de Finetti state will be not only non-signalling but also quantum.
Due to the general structure of the de Finetti reductions and the increasing use of conditional probability distributions in quantum information theory, we also hope that the presented reductions will be useful in other applications apart from cryptography, such as quantum tomography, as was the case for the quantum post selection theorem~\cite{christandl2009postselection}.

The techniques used to prove our theorems (mainly combinatoric arguments) differ from the techniques used in previous papers to establish general de Finetti theorems. We therefore hope that they will shed new light on de Finetti reductions in general. For example, it might be possible to apply some ideas from the proof in (device dependent) quantum de Finetti reductions. 

\begin{acknowledgments}
The authors thank Roger Colbeck and Michael Walter for discussing a preliminary version of this work. This work was supported by the Swiss National Science Foundation (via the National Centre of Competence in Research ``QSIT'' and SNF project No. 200020-135048), by the European Research Council (via project No. 258932), by the CHIST-ERA project ``DIQIP'' and by the EC STREP project ``RAQUEL''. 
\end{acknowledgments}


\appendix
\section{A direct proof of the de Finetti reduction for states with CHSH symmetry \label{sec:chsh-proof}}

Corollary \ref{cor:chsh-post-selection} is a de Finetti reduction, specialised for states with the CHSH symmetry, as defined in Definition \ref{def:chsh-symmetry}, and therefore relevant for device independent cryptographic protocols that are based on the violation of the CHSH inequality \cite{CHSH}. In Section \ref{sec:diamond_app} we show how to apply this de Finetti reduction to simplify the analysis of such protocols. In this section we prove the corollary directly, without using the general de Finetti reduction, given in Theorem \ref{thm:post-selection}. In Section \ref{sec:cor_derivation} we also show how to derive the corollary from the general theorem.

In order to prove Corollary \ref{cor:chsh-post-selection} we construct a specific de Finetti state $\tau^\mathcal{CHSH}_{AB|XY}$ for which the corollary holds. As our de Finetti state we choose $\tau^\mathcal{CHSH}_{AB|XY}=\int Q_{A_1B_1|X_1Y_1}^{\otimes n} \mathrm{d}Q_{A_1B_1|X_1Y_1}$ to be a convex combination of states $Q_{A_1B_1|X_1Y_1}^{\otimes n}$ where $Q_{A_1B_1|X_1Y_1}$ is the basic state given in Figure \ref{fig:CHSH_symmetry}. As a density measure we choose $\mathrm{d}Q$ to be uniform over all states $Q_{A_1B_1|X_1Y_1}$ of this form, i.e., we integrate uniformly over different values of $p\in [0,\frac{1}{2}]$. 

We can write $\tau^\mathcal{CHSH}_{AB|XY}$ explicitly by using the following notation. For every $a,b,x,y\in \{0,1\}^n$ denote by $N_{CHSH}$ the number of indices $m\in[n]$ for which the foursome $(a_m,b_m,x_m,y_m)$ fulfils the CHSH condition, i.e. $a_m \oplus b_m = x_m \cdot y_m$. $n - N_{CHSH}$ is then the number of indices for which the foursome $(a_m,b_m,x_m,y_m)$ does not fulfil the CHSH condition.  We formally define $\tau^\mathcal{CHSH}_{AB|XY}$:

\begin{defn}
	$\tau^\mathcal{CHSH}_{AB|XY}$ is the non-signalling state
	\begin{align*}
		\tau^\mathcal{CHSH}_{AB|XY}(ab|xy) & =  \int Q_{A_1B_1|X_1Y_1}^{\otimes n}(ab|xy)\: \mathrm{dQ}_{A_1B_1|X_1Y_1}\\
		& =\int_0^{\frac{1}{2}} \left( \frac{1}{2}-p \right)^{N_{CHSH}} p^{\left( n-N_{CHSH} \right)} 2\mathrm{d}p \; .
	\end{align*}
\end{defn}
The de Finetti state is non-signalling since the states $Q_{A_1B_1|X_1Y_1}$ are non-signalling for every value of $p$ (see Figure~\ref{fig:CHSH_symmetry}).

\begin{lem}\label{lem:chsh-de-finetti-bound}
	For all $a,b,x,y\in\{0,1\}^n$, 
	\[
		\tau^\mathcal{CHSH}_{AB|XY}(ab|xy) = 2^{-n}{n \choose N_{CHSH}}^{-1} \frac{1}{(n+1)} \;.
	\]
\end{lem}
\begin{proof}
	The integral above can be solved explicitly: 
	\begin{align*}
		\tau^\mathcal{CHSH}_{AB|XY}(ab|xy) & =  \int_0^{\frac{1}{2}} \left( \frac{1}{2}-p \right)^{N_{CHSH}} p^{\left( n-N_{CHSH} \right)} 2\mathrm{d}p\\
		&=2^{-n} \int_0^1 (1-q)^{N_{CHSH}} q^{\left( n-N_{CHSH} \right)} \mathrm{d}q  \\
		&=2^{-n} \mathrm{B}(n-N_{CHSH}+1,N_{CHSH}+1) \\
		&=2^{-n} {n \choose N_{CHSH}}^{-1} \frac{1}{(n+1)} \;. 
	\end{align*} 
	where B is the Beta function. Recall that $N_{CHSH}$ is a functions of $a,b,x$ and $y$ although we do not write it explicitly.
\end{proof}

The following lemma gives an upper bound on any entry $\mathrm{P}_{AB|XY}(ab|xy)$ of every permutation invariant state $\mathrm{P}_{AB|XY}$ with the CHSH symmetry.
\begin{lem}\label{lem:upper-bound-permutation}
	For every permutation invariant state $\mathrm{P}_{AB|XY}$ with the CHSH symmetry and $a,b,x,y\in\{0,1\}^n$
	\[
		\mathrm{P}_{AB|XY}(ab|xy) \leq  2^{-n}{n \choose N_{CHSH}}^{-1}  \;. 
	\]
\end{lem}
The idea behind the proof of this lemma is to bound the value of a specific entry $\mathrm{P}_{AB|XY}(ab|xy)$ by counting how many entries $\mathrm{P}_{AB|XY}(\tilde{a}\tilde{b}|xy)$ must have the same value as  $\mathrm{P}_{AB|XY}(ab|xy)$ due to the symmetry of $\mathrm{P}_{AB|XY}$. Since the sum of all entries with particular inputs $x,y$ is 1 a bound on $\mathrm{P}_{AB|XY}(ab|xy)$ follows.

\begin{proof}
Given $a,b,x,y$ imagine that we are placing a colored ball above each foursome $(a_i,b_i,x_i,y_i)$ as in Figure~\ref{fig:balls}. If the foursome fulfils the CHSH condition we label it with a blue ball, otherwise with a red ball. With this picture in mind, the CHSH symmetry as in Definition~\ref{def:chsh-symmetry} says that by changing two balls of the same color we do not change the value according to the probability distribution $\mathrm{P}_{AB|XY}$.

\begin{figure}
\begin{centering}
	\begin{tikzpicture}[dot/.style={circle,inner sep =2pt}]
 	
	\draw (0,1.5) node {$a:$};
	\draw (0,1) node {$b:$};
	\draw (0,0.5) node {$x:$};
	\draw (0,0) node {$y:$};
	\draw[font=\scriptsize] (0, -0.55) node {$i:$};

	\node[dot, fill=blue] at (0.75,2) {};
	\draw (0.75,1.5) node {$0$};
	\draw (0.75,1) node {$0$};
	\draw (0.75,0.5) node {$0$};
	\draw (0.75,0) node {$0$};
	\draw[font=\scriptsize] (0.75, -0.55) node {$1$};
		
	\node[dot, fill=red] at (1.15,2) {};
	\draw (1.15,1.5) node {$1$};
	\draw (1.15,1) node {$0$};
	\draw (1.15,0.5) node {$0$};
	\draw (1.15,0) node {$1$};
	\draw[font=\scriptsize] (1.15, -0.55) node {$2$};
	
	\node[dot, fill=blue] at (1.55,2) {};
	\draw (1.55,1.5) node {$0$};
	\draw (1.55,1) node {$1$};
	\draw (1.55,0.5) node {$1$};
	\draw (1.55,0) node {$1$};
	\draw[font=\scriptsize] (1.55, -0.55) node {$3$};

	\node[dot, fill=red] at (1.95,2) {};
	\draw (1.95,1.5) node {$1$};
	\draw (1.95,1) node {$0$};
	\draw (1.95,0.5) node {$1$};
	\draw (1.95,0) node {$0$};
	\draw[font=\scriptsize] (1.95, -0.55) node {$4$};
	
	\node[dot, fill=red] at (2.35,2) {};
	\draw (2.35,1.5) node {$1$};
	\draw (2.35,1) node {$1$};
	\draw (2.35,0.5) node {$1$};
	\draw (2.35,0) node {$1$};
	\draw[font=\scriptsize] (2.35, -0.55) node {$5$};
	
	\node[dot, fill=blue] at (2.75,2) {};
	\draw (2.75,1.5) node {$0$};
	\draw (2.75,1) node {$0$};
	\draw (2.75,0.5) node {$0$};
	\draw (2.75,0) node {$1$};
	\draw[font=\scriptsize] (2.75, -0.55) node {$6$};
	
	\node[dot, fill=blue] at (3.15,2) {};
	\draw (3.15,1.5) node {$1$};
	\draw (3.15,1) node {$1$};
	\draw (3.15,0.5) node {$0$};
	\draw (3.15,0) node {$0$};
	\draw[font=\scriptsize] (3.15, -0.55) node {$7$};

	\node[dot, fill=red] at (3.55,2) {};
	\draw (3.55,1.5) node {$0$};
	\draw (3.55,1) node {$1$};
	\draw (3.55,0.5) node {$0$};
	\draw (3.55,0) node {$1$};
	\draw[font=\scriptsize] (3.55, -0.55) node {$8$};
	
	\end{tikzpicture}
\par\end{centering}
\caption{Partition to CHSH quartets. If the foursome $(a_i,b_i,x_i,y_i)$  fulfills the CHSH condition it is denoted by a blue ball, otherwise by a red ball. \label{fig:balls}}
\end{figure}

\begin{figure*}
\begin{centering}
	\begin{tikzpicture}[dot/.style={circle,inner sep =2pt}]
 	
	\begin{scope}
		\draw (0,1.5) node {$\pi(a):$};
		\draw (0,1) node {$\pi(b):$};
		\draw (0,0.5) node {$\pi(x):$};
		\draw (0,0) node {$\pi(y):$};
		\draw[font=\scriptsize] (0, -0.55) node {$\pi(i):$};
	
		\node[dot, fill=red] at (0.75,2) {};
		\draw (0.75,1.5) node {$$};
		\draw (0.75,1) node {$$};
		\draw (0.75,0.5) node {$$};
		\draw (0.75,0) node {$$};
		\draw[font=\scriptsize] (0.75, -0.55) node {$$};
			
		\node[dot, fill=blue] at (1.15,2) {};
		\draw (1.15,1.5) node {$$};
		\draw (1.15,1) node {$$};
		\draw (1.15,0.5) node {$$};
		\draw (1.15,0) node {$$};
		\draw[font=\scriptsize] (1.15, -0.55) node {$$};
		
		\node[dot, fill=blue] at (1.55,2) {};
		\draw (1.55,1.5) node {$$};
		\draw (1.55,1) node {$$};
		\draw (1.55,0.5) node {$$};
		\draw (1.55,0) node {$$};
		\draw[font=\scriptsize] (1.55, -0.55) node {$$};
	
		\node[dot, fill=blue] at (1.95,2) {};
		\draw (1.95,1.5) node {$$};
		\draw (1.95,1) node {$$};
		\draw (1.95,0.5) node {$$};
		\draw (1.95,0) node {$$};
		\draw[font=\scriptsize] (1.95, -0.55) node {$$};
		
		\node[dot, fill=red] at (2.35,2) {};
		\draw (2.35,1.5) node {$$};
		\draw (2.35,1) node {$$};
		\draw (2.35,0.5) node {$$};
		\draw (2.35,0) node {$$};
		\draw[font=\scriptsize] (2.35, -0.55) node {$$};
		
		\node[dot, fill=blue] at (2.75,2) {};
		\draw (2.75,1.5) node {$$};
		\draw (2.75,1) node {$$};
		\draw (2.75,0.5) node {$$};
		\draw (2.75,0) node {$$};
		\draw[font=\scriptsize] (2.75, -0.55) node {$$};
		
		\node[dot, fill=red] at (3.15,2) {};
		\draw (3.15,1.5) node {$$};
		\draw (3.15,1) node {$$};
		\draw (3.15,0.5) node {$$};
		\draw (3.15,0) node {$$};
		\draw[font=\scriptsize] (3.15, -0.55) node {$$};
	
		\node[dot, fill=red] at (3.55,2) {};
		\draw (3.55,1.5) node {$$};
		\draw (3.55,1) node {$$};
		\draw (3.55,0.5) node {$$};
		\draw (3.55,0) node {$$};
		\draw[font=\scriptsize] (3.55, -0.55) node {$$};
	\end{scope}
	
	\begin{scope}[shift={(5,0)}]
		\draw (0,1.5) node {$\pi(a):$};
		\draw (0,1) node {$\pi(b):$};
		\draw (0,0.5) node {$\pi(x):$};
		\draw (0,0) node {$\pi(y):$};
		\draw[font=\scriptsize] (0, -0.55) node {$\pi(i):$};
	
		\node[dot, fill=red] at (0.75,2) {};
		\draw (0.75,1.5) node {$$};
		\draw (0.75,1) node {$$};
		\draw (0.75,0.5) node {$$};
		\draw (0.75,0) node {$$};
		\draw[font=\scriptsize] (0.75, -0.55) node {$$};
			
		\node[dot, fill=blue] at (1.15,2) {};
		\draw (1.15,1.5) node {$$};
		\draw (1.15,1) node {$$};
		\draw (1.15,0.5) node {$$};
		\draw (1.15,0) node {$$};
		\draw[font=\scriptsize] (1.15, -0.55) node {$$};
		
		\node[dot, fill=blue] at (1.55,2) {};
		\draw (1.55,1.5) node {$0$};
		\draw (1.55,1) node {$1$};
		\draw (1.55,0.5) node {$1$};
		\draw (1.55,0) node {$1$};
		\draw[font=\scriptsize] (1.55, -0.55) node {$3$};
	
		\node[dot, fill=blue] at (1.95,2) {};
		\draw (1.95,1.5) node {$$};
		\draw (1.95,1) node {$$};
		\draw (1.95,0.5) node {$$};
		\draw (1.95,0) node {$$};
		\draw[font=\scriptsize] (1.95, -0.55) node {$$};
		
		\node[dot, fill=red] at (2.35,2) {};
		\draw (2.35,1.5) node {$1$};
		\draw (2.35,1) node {$1$};
		\draw (2.35,0.5) node {$1$};
		\draw (2.35,0) node {$1$};
		\draw[font=\scriptsize] (2.35, -0.55) node {$5$};
		
		\node[dot, fill=blue] at (2.75,2) {};
		\draw (2.75,1.5) node {$0$};
		\draw (2.75,1) node {$0$};
		\draw (2.75,0.5) node {$0$};
		\draw (2.75,0) node {$1$};
		\draw[font=\scriptsize] (2.75, -0.55) node {$6$};
		
		\node[dot, fill=red] at (3.15,2) {};
		\draw (3.15,1.5) node {$$};
		\draw (3.15,1) node {$$};
		\draw (3.15,0.5) node {$$};
		\draw (3.15,0) node {$$};
		\draw[font=\scriptsize] (3.15, -0.55) node {$$};
	
		\node[dot, fill=red] at (3.55,2) {};
		\draw (3.55,1.5) node {$0$};
		\draw (3.55,1) node {$1$};
		\draw (3.55,0.5) node {$0$};
		\draw (3.55,0) node {$1$};
		\draw[font=\scriptsize] (3.55, -0.55) node {$8$};
	\end{scope}

	\begin{scope}[shift={(10,0)}]
		\draw (0,1.5) node {$\pi(a):$};
		\draw (0,1) node {$\pi(b):$};
		\draw (0,0.5) node {$\pi(x):$};
		\draw (0,0) node {$\pi(y):$};
		\draw[font=\scriptsize] (0, -0.55) node {$\pi(i):$};
	
		\node[dot, fill=red] at (0.75,2) {};
		\draw (0.75,1.5) node {$1$};
		\draw (0.75,1) node {$0$};
		\draw (0.75,0.5) node {$0$};
		\draw (0.75,0) node {$1$};
		\draw[font=\scriptsize, orange] (0.75, -0.55) node {$2$};
			
		\node[dot, fill=blue] at (1.15,2) {};
		\draw (1.15,1.5) node {$0$};
		\draw (1.15,1) node {$0$};
		\draw (1.15,0.5) node {$0$};
		\draw (1.15,0) node {$0$};
		\draw[font=\scriptsize, orange] (1.15, -0.55) node {$1$};
		
		\node[dot, fill=blue] at (1.55,2) {};
		\draw (1.55,1.5) node {$0$};
		\draw (1.55,1) node {$1$};
		\draw (1.55,0.5) node {$1$};
		\draw (1.55,0) node {$1$};
		\draw[font=\scriptsize] (1.55, -0.55) node {$3$};
	
		\node[dot, fill=blue] at (1.95,2) {};
		\draw (1.95,1.5) node {$1$};
		\draw (1.95,1) node {$1$};
		\draw (1.95,0.5) node {$0$};
		\draw (1.95,0) node {$0$};
		\draw[font=\scriptsize, purple] (1.95, -0.55) node {$7$};
		
		\node[dot, fill=red] at (2.35,2) {};
		\draw (2.35,1.5) node {$1$};
		\draw (2.35,1) node {$1$};
		\draw (2.35,0.5) node {$1$};
		\draw (2.35,0) node {$1$};
		\draw[font=\scriptsize] (2.35, -0.55) node {$5$};
		
		\node[dot, fill=blue] at (2.75,2) {};
		\draw (2.75,1.5) node {$0$};
		\draw (2.75,1) node {$0$};
		\draw (2.75,0.5) node {$0$};
		\draw (2.75,0) node {$1$};
		\draw[font=\scriptsize] (2.75, -0.55) node {$6$};
		
		\node[dot, fill=red] at (3.15,2) {};
		\draw (3.15,1.5) node {$1$};
		\draw (3.15,1) node {$0$};
		\draw (3.15,0.5) node {$1$};
		\draw (3.15,0) node {$0$};
		\draw[font=\scriptsize, purple] (3.15, -0.55) node {$4$};
	
		\node[dot, fill=red] at (3.55,2) {};
		\draw (3.55,1.5) node {$0$};
		\draw (3.55,1) node {$1$};
		\draw (3.55,0.5) node {$0$};
		\draw (3.55,0) node {$1$};
		\draw[font=\scriptsize] (3.55, -0.55) node {$8$};
	\end{scope}

	\end{tikzpicture}

\par\end{centering}
\caption{The permutation $\pi$. \label{fig:permutation}}
\end{figure*}

Given a specific entry $\mathrm{P}_{AB|XY}(ab|xy)$ we would like to know how many entries with the same inputs $x,y$ have to have the same value as the given entry. Formally, we would like to have a lower bound on 
\begin{align*}
	\mathcal{N}(a,b,x,y) =  \Big| &\{ (\tilde{a},\tilde{b}) \in \{0,1\}^n \times \{0,1\}^n \quad \text{s.t.} \\
	&\mathrm{P}_{AB|XY}(\tilde{a}\tilde{b}|xy) = \mathrm{P}_{AB|XY}(ab|xy) \} \Big| \; .
\end{align*}
How small can $\mathcal{N}(a,b,x,y)$ be? Or in other words, in how many ways can we change $a$ and $b$ while getting an entry $\mathrm{P}_{AB|XY}(\tilde{a}\tilde{b}|xy)$ with the same value?  We now prove 
\[
	\mathcal{N}(a,b,x,y) \geq 2^{n} {n \choose {N_{CHSH}} }\;.
\]
One way of changing $(a,b,x,y)$ to $(\tilde{a},\tilde{b},x,y)$ without changing the value of the entry is to change $(a,b,x,y)$ to $(\tilde{a},\tilde{b},x,y)$ such that both will have the same sequence of colored balls. For example, in Figure \ref{fig:balls} we can change $(a_1,b_1,x_1,y_1)=(0,0,0,0)$ to $(\tilde{a}_1,\tilde{b}_1,x_1,y_1)=(1,1,0,0)$ since they have the same inputs $(x_1,y_1)=(0,0)$ and both will be denoted by a blue ball (therefore according to the symmetry this change will not affect the overall value of the entry). In how many such different ways can we change $a$ and $b$? For every index $i\in[n]$ and every input bits $x_i,y_i$ there are two $a_i,b_i$ for which the CHSH conditions holds (i.e. blue ball) and two for which it does not (red ball). Therefore there are exactly $2^n$ different pairs of strings $(\tilde{a},\tilde{b})$ such that $(a,b,x,y)$ and $(\tilde{a},\tilde{b},x,y)$ have the same sequence of colored balls and therefore $\mathrm{P}_{AB|XY}(\tilde{a}\tilde{b}|xy) = \mathrm{P}_{AB|XY}(ab|xy)$.

Changing $a$  and $b$ in different ways than the way given above will necessarily change the colors sequence. However, we can still prove using the permutation invariance of $\mathrm{P}_{AB|XY}$ that for some specific changes the value of the entry will still stay the same. The specific changes that we consider are determined by permutations of the colored balls. 

In order to understand how every different permutation of the balls $\pi$ is realised as a permutation on $x,y,a,b$ consider the example drawn in Figure \ref{fig:permutation}. On the left side we see a permutation of the balls from Figure \ref{fig:balls}. We start by filling up the columns for which there is no change in the color of the ball with the original columns. Then we pair the permuted balls such that each blue ball is replaced with a red ball, and permute the columns according to this paring. The permutation in the figure, for example, is just the permutation of indices $(1,2)$ and $(4,7)$. In general, every permutation of the balls can be described by such pairing and between every two different permutation we will have at least one index which will be permuted in one of them and not in the other. 

For every permutation $\pi$ as described above we have 
\begin{equation}\label{eq:perm_result}
	\mathrm{P}_{AB|XY}(ab|xy) = \mathrm{P}_{AB|XY}\left( \pi(a)\pi(b)|\pi(x)\pi(y) \right) 
\end{equation}
since the state is permutation invariant. We will now show that due to the CHSH symmetry we also have
\begin{equation}\label{eq:sym_result}
	\mathrm{P}_{AB|XY}\left( \pi(a)\pi(b)|\pi(x)\pi(y) \right)  = \mathrm{P}_{AB|XY}\left(\tilde{a}_{\pi}\tilde{b}_{\pi}|xy \right) 
\end{equation}
where $\tilde{a}_{\pi}=a$ and $\tilde{b}_{\pi}$ is derived from $b$ by negation of all the bits which are being permuted in $\pi$. 

To see that Equation (\ref{eq:sym_result}) is correct recall that the permutation $\pi$ permuted two columns $i,j$ only if for one of them the CHSH condition holds and for the other not.  Therefore, if for example we had $a_i \oplus b_i = x_i \cdot y_i$ (and the index $i$ was permuted by $\pi$) then $\left(\pi(a)\right)_i \oplus \left(\pi(b)\right)_i \neq \left(\pi(x)\right)_i \cdot \left(\pi(y)\right)_i$. By definition $(\tilde{a}_{\pi})_i = a_i$ and $(\tilde{b}_{\pi})_i = \overline{b_i}$ and therefore we also know that $a_i \oplus (\tilde{b}_{\pi})_i \neq x_i \cdot y_i$. Combining this with the CHSH symmetry and proceeding in the same way for all the indices that $\pi$ permutes, we get Equation (\ref{eq:sym_result}).

Combining Equations (\ref{eq:perm_result}) and (\ref{eq:sym_result}) we get
\[
	\mathrm{P}_{AB|XY}(ab|xy) = \mathrm{P}_{AB|XY}\left(a \tilde{b}_{\pi}|xy \right) \;.
\]
Any different permutation $\pi$ will result in a different $\tilde{b}_{\pi}$ and therefore for any different permutation $\pi$ we have a different entry $\mathrm{P}_{AB|XY}(\tilde{a},\tilde{b},x,y)$ with the same value as the original entry $\mathrm{P}_{AB|XY}(a,b,x,y)$. Since there are ${n \choose N_{CHSH}}$ different permutations of the balls we have ${n \choose N_{CHSH}}$  different ways of changing $(a,b,x,y)$ to $(\tilde{a},\tilde{b},x,y)$. 

We can now answer our original question and bound $\mathcal{N}(a,b,x,y)$. We can combine both of the ways given above to change $a$ and $b$ without changing the value of the entry according to $\mathrm{P}_{AB|XY}$ (with or without changing the colors sequence). This implies that in total there are at least $2^{n} \times {n \choose {N_{CHSH}}}$ different ways of changing $a$  and $b$ and we can conclude that
\begin{equation} \label{eq:N_bound}
	\mathcal{N}(a,b,x,y) \geq 2^{n}  {n \choose {N_{CHSH}}} \;.
\end{equation}

Since for all $x,y$ $\sum_{a,b} \mathrm{P}_{AB|XY}(ab|xy) = 1$, we get from Equation \eqref{eq:N_bound} the following bound on the value of $\mathrm{P}_{AB|XY}(ab|xy)$:
\[
	\mathrm{P}_{AB|XY}(ab|xy) \leq 2^{-n} {n \choose {N_{CHSH}}}^{-1} \;. \qedhere
\]
\end{proof}

We can now prove Corollary \ref{cor:chsh-post-selection} directly. 
\begin{proof}[A direct proof of Corollary \ref{cor:chsh-post-selection}]
	By combining Lemma \ref{lem:chsh-de-finetti-bound} and Lemma \ref{lem:upper-bound-permutation} we get Corollary \ref{cor:chsh-post-selection}.
\end{proof}

The above proof for states which have the CHSH symmetry can be also applied to states which have the symmetry induced by the more general chained Bell inequalities \cite{braunstein1990wringing,Barrett2006Maximally} in a similar way. Since the number of measurements of the basic states $Q_{A_1B_1|X_1Y_1}$ does not play a role in the structure of our de Finetti state (see Lemma \ref{lem:chsh-de-finetti-bound}) the same bounds exactly will hold for states with the chained Bell inequality symmetry.


\section{Proof of the general de Finetti reduction \label{sec:general-proof}}

In this section we prove our most general de Finetti reduction given in Theorem \ref{thm:post-selection}.

The proof proceeds along the same lines as the direct proof of Corollary \ref{cor:chsh-post-selection}  in Appendix \ref{sec:chsh-proof}. We start by explaining the types of symmetries $\mathcal{S}$ that we deal with and how to construct the appropriate de Finetti state $\tau^{\mathcal{S}}_{A|X}$. We then give a lower bound on the entries of the de Finetti state, analogously to Lemma \ref{lem:chsh-de-finetti-bound}, and an upper bound on the entries of a permutation invariant state $\mathrm{P}_{A|X}$ with the symmetry $\mathcal{S}$, analogously to Lemma \ref{lem:upper-bound-permutation}. Using these two bounds we get Theorem \ref{thm:post-selection}.

\subsection{Symmetries and de Finetti states}

A symmetry $\mathcal{S}$ is a set of conditions. We say that a state $\mathrm{P}_{A|X}$ has a symmetry $\mathcal{S}$ if it fulfils all of these conditions. 

For any symmetry $\mathcal{S}$ that we consider we define a different de Finetti state $\tau^{\mathcal{S}}_{A|X}$ of the form $\int Q_{A_1|X_1}^{\otimes n} \mathrm{d}Q_{A_1|X_1}$. When defining such a de Finetti state for a specific type of symmetry $\mathcal{S}$ we are free to choose  the measure $\mathrm{d}Q_{A_1|X_1}$ as we like. The key idea is to choose the structure of the states $Q_{A_1|X_1}$ on which we integrate in such a way that it ``encodes'' the symmetry $\mathcal{S}$ that we consider.

For example, assume we consider a family of states $\mathrm{P}_{A|X}$, with $a,x \in \{0,1\}^n$, which has the following type of symmetry $\mathcal{S}$:
\[
	\forall i\in[n] \;\; \forall a_{\overline{i}},x \quad 
	\mathrm{P}_{A|X}(a_{\overline{i}}0,|x) = \mathrm{P}_{A|X}(a_{\overline{i}}1|x)
\]
(that is, given $a_{\overline{i}}=a_1,\dotsc,a_{i-1},a_{i+1},\dotsc,a_n$ and $x$ the probability that the i'th bit $a_i$ will be 0 or 1 is the same).  We then want $Q_{A_1|X_1}$ to have the following property:
\[
	\forall x_1\in\{0,1\} \quad Q_{A_1|X_1}(0|x_1) = Q_{A_1|X_1}(1|x_1)
\]
and we say that $Q_{A_1|X_1}$ encodes the symmetry $\mathcal{S}$. 

For the more general treatment it will be easier to start by defining the allowed structure of the state  $Q_{A_1|X_1}$  and from it deduce the different types of symmetries and states that we consider. 

\subsubsection*{Allowed $Q_{A_1|X_1}$ states}

Consider a 1-party state $Q_{A_1|X_1}$ where $x_1\in\{0,1,\dotsc,m-1\}$ and $a_1\in\{0,1,\dotsc l-1 \}$. We can think about $Q_{A_1|X_1}$ as $m$ vectors of size $l$. We call each of the $l$ long vectors an \emph{input vector}, since it describes the probability distribution of the outputs, given a specific input (see Figure \ref{fig:input_vector}). Defining a state  $Q_{A_1|X_1}$ then reduces to defining its input vectors. 

\begin{figure}
\begin{centering}
\begin{tikzpicture}

	\draw[step=1] (0,0) grid (6,1);
	\draw[thick] (3,-0.25) -- (3,1.25);
	
	\draw[font=\small] (0.5,0.5) node {$Q(0|0)$};
	\draw[font=\small] (1.5,0.5) node {$Q(1|0)$};
	\draw[font=\small] (2.5,0.5) node {$Q(2|0)$};
	
	\draw[font=\small] (3.5,0.5) node {$Q(0|1)$};
	\draw[font=\small] (4.5,0.5) node {$Q(1|1)$};
	\draw[font=\small] (5.5,0.5) node {$Q(2|1)$};
	
	\draw[decorate,decoration={brace,mirror,amplitude=10pt}] (0,-0.3) -- (2.95,-0.3) node [midway,yshift=-0.6cm] {input vector};

\end{tikzpicture}

\par\end{centering}
\caption{A ststem $Q_{A_1|X_1}$ with $m=2$ inputs and $l=3$ outputs}
\label{fig:input_vector}
\end{figure}
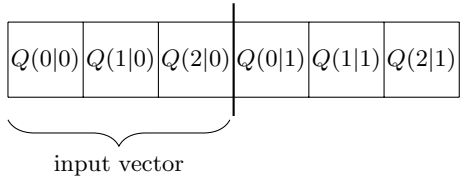

Keeping in mind that we will need to integrate over $Q_{A_1|X_1}$ to get a de Finetti state, we fill in the entries $Q_{A_1|X_1}(a|x)$ of the input vectors with different parameters $\{ p_1, p_2, \dotsc p_{d} \}$, while making sure that the sum of the input vector is 1 for every value of the parameters $p_i$. The number of parameters $d$ that we use to define  $Q_{A_1|X_1}$ quantifies the number of \emph{degrees of freedom} that  $Q_{A_1|X_1}$ has, and it is bounded by $(l-1)m$. Figure \ref{fig:paramaters_example} shows two different ways of filling in an input vector of length 3 with parameters. 

\begin{figure}
\begin{centering}
\begin{tikzpicture}

	\begin{scope}[shift={(0,0)}]
		\draw[step=1] (0,0) grid (3,1);
		
		\draw[red] (0.5,0.5) node {$p_1$};
		\draw[blue] (1.5,0.5) node {$p_2$};
		\draw[font=\small] (2.25,0.75) node {$1-$};
		\draw[font=\small] (2.5,0.5) node {$p_1-$};
		\draw[font=\small] (2.75,0.25) node {$p_2$};
			
		\draw (1.5,1.5) node {$d = 2$};
	\end{scope}

	\begin{scope}[shift={(5,0)}]
		\draw[step=1] (0,0) grid (3,1);
		
		\draw[red] (0.5,0.5) node {$p_1$};
		\draw[font=\small] (1.35,0.65) node {$1-$};
		\draw[font=\small] (1.65,0.35) node {$2p_1$};
		\draw[red] (2.5,0.5) node {$p_1$};
				
		\draw (1.5,1.5) node {$d = 1$};
	\end{scope}

\end{tikzpicture}
\par\end{centering}
\caption{Different ways of filling an input vector of length 3 with $d$ paramaters}
\label{fig:paramaters_example}
\end{figure}
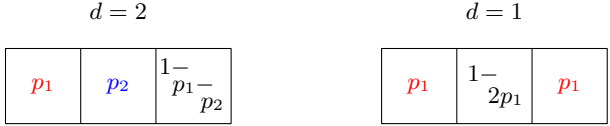

We can now define a specific set of allowed states $Q_{A_1|X_1}$. 
\begin{defn} \label{defn:allowed_Q}
	A state $Q_{A_1|X_1}$ is said to be allowed if given any two of its input vectors, they are either a permutation of one another or they have a completely different set of parameters. 
\end{defn}

In Figure \ref{fig:basic_state_example} we give 2 examples for allowed $Q_{A_1|X_1}$ states with 2 inputs and 3 outputs. 

\begin{figure}
\begin{centering}
\begin{tikzpicture}

	\begin{scope}
		\draw[step=1] (0,0) grid (6,1);
		\draw[thick] (3,-0.25) -- (3,1.25);
		
		\draw[red] (0.5,0.5) node {$p_1$};
		\draw[blue] (1.5,0.5) node {$p_2$};
		\draw[font=\small] (2.25,0.75) node {$1-$};
		\draw[font=\small] (2.5,0.5) node {$p_1-$};
		\draw[font=\small] (2.75,0.25) node {$p_2$};
		
		\draw[blue] (3.5,0.5) node {$p_2$};
		\draw[font=\small] (4.25,0.75) node {$1-$};
		\draw[font=\small] (4.5,0.5) node {$p_1-$};
		\draw[font=\small] (4.75,0.25) node {$p_2$};
		\draw[red] (5.5,0.5) node {$p_1$};	
	\end{scope}

	\begin{scope}[shift={(0,2)}]
		\draw[step=1] (0,0) grid (6,1);
		\draw[thick] (3,-0.25) -- (3,1.25);
		
		\draw[red] (0.5,0.5) node {$p_1$};
		\draw[red] (1.5,0.5) node {$p_1$};
		\draw[font=\small] (2.35,0.65) node {$1-$};
		\draw[font=\small] (2.65,0.35) node {$2p_1$};
		
		\draw[blue] (3.5,0.5) node {$p_2$};
		\draw[font=\small] (4.35,0.65) node {$1-$};
		\draw[font=\small] (4.65,0.35) node {$2p_2$};
		\draw[blue] (5.5,0.5) node {$p_2$};
	\end{scope}

\end{tikzpicture}

\par\end{centering}
\caption{Two examples for allowed $Q_{A_1|X_1}$ states. Above, in each input vector, we have a different independent parameter. Below, the two input vectors are permutations of one another.}
\label{fig:basic_state_example}
\end{figure}
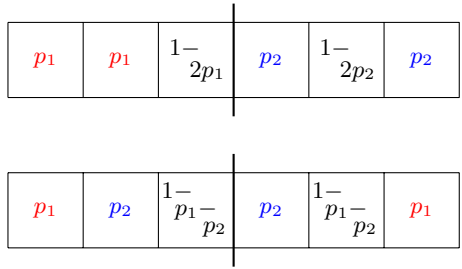

\subsubsection*{The symmetry $\mathcal{S}$ behind $Q_{A_1|X_1}$}

When considering a specific state $Q_{A_1|X_1}$ it is easy to say which set of conditions it fulfils, i.e., which symmetry $\mathcal{S}$ it encodes. For example, the state on the top of Figure~\ref{fig:basic_state_example} encodes the following symmetry of a state $\mathrm{P}_{A|X}$ with $a\in\{0,1,2\}^n$ and $x\in \{0,1\}^n$,
\begin{align*}
	\forall i\in[n] \quad \forall x_{\overline{i}},a_{\overline{i}} \quad & \mathrm{P}_{A|X}(a_{\overline{i}}0|x_{\overline{i}}0) = \mathrm{P}_{A|X}(a_{\overline{i}}1|x_{\overline{i}}0) \\
													 & \mathrm{P}_{A|X}(a_{\overline{i}}0|x_{\overline{i}}1) = \mathrm{P}_{A|X}(a_{\overline{i}}2|x_{\overline{i}}1)  \;.
\end{align*}

More generally, the symmetry $\mathcal{S}$ can be constructed from $Q_{A_1|X_1}$ as follows. 
\begin{defn} \label{defn:symmetry_relation}
	Given a state $Q_{A_1|X_1}$ as above the symmetry $\mathcal{S}$ is defined by the following symmetry conditions:
	
	For all $i\in[n]$,  for all $x,a$ and $x',a'$ where $a'=a_1\dots a_{i-1} a'_i a_{i+1} \dots a_n$ and $x'$ is defined in a similar way, if $Q_{A_1|X_1}(a_i|x_i)=Q_{A_1|X_1}(a'_i|x'_i)$ then $\mathrm{P}_{A|X}(a|x)=\mathrm{P}_{A|X}(a'|x')$.
\end{defn}
That is, if we change the pair $(a_i,x_i)$ to some  $(a'_i,x'_i)$ of the ``same type'' according to  $Q_{A_1|X_1}$, then the probability according to $\mathrm{P}_{A|X}$ does not change. 

In Definition \ref{defn:symmetry_relation} we started from the state $Q_{A_1|X_1}$ and derived the symmetry $\mathcal{S}$. However, given a set of conditions $\mathcal{S}$ one can also try to construct a state $Q_{A_1|X_1}$ which fulfils them. For every symmetry $\mathcal{S}$ for which a state $Q_{A_1|X_1}$ can be constructed such that the condition in Definition \ref{defn:symmetry_relation} holds our proof can be applied. In other words, the only thing needed for our theorem to hold is a pair $\left( \mathcal{S}, Q_{A_1|X_1}\right)$ with the desired relationship.

\subsubsection*{The de Finetti state --- integration over $Q_{A_1|X_1}$}

Given a specific structure of  $Q_{A_1|X_1}$ as previously described, we can now perform the integration over a tensor product of $n$ such states and get a de Finetti  state.  As mentioned before, we are free to choose the measure $\mathrm{d}Q_{A_1|X_1}$ with which we perform the integration. 

For simplicity, we only consider $Q_{A_1|X_1}$ states in which all the input vectors are permutations of one another (recall Definition \ref{defn:allowed_Q}). It will later become clear, that if we have more independent input vectors, then we can just multiply the different integrals by one another. In general, due to our proof technique, our entire proof can be applied independently for each set of permuted input vectors and then combined in the end to one proof by multiplying the results. 

In the rest of the proof, we use the following notation. Given the state $Q_{A_1|X_1}$ we denote by $t_i$, for $0 \leq i \leq d$,  the number of times the parameter $p_i$ appears in each input vector of  $Q_{A_1|X_1}$. In addition, we define $t_{d+1}$ to be the number of times the ``unfree'' entry appears in the input vector. Using this notation, we can set the range of the parameter $p_i$ to be $\left[ 0,c_i \equiv \frac{1}{t_i}\left(1 - \sum_{j<i}t_jp_j \right) \right]$. 
As an example, consider the input vector in Figure \ref{fig:input_vector_t_and_c}. $p_1$ appears two times and therefore $t_1=2$ and $c_1=\frac{1}{2}$. Indeed, in order for this input vector to be a valid probability distribution we must have $p_1\in \left[ 0,\frac{1}{2} \right]$. For $p_2$ we have $t_2=3$ and $c_2=\frac{1}{3}\left(1-2p_1 \right)$, and $t_3=1$. 

\begin{figure}
\begin{centering}
\begin{tikzpicture}

	\draw[step=1] (0,0) grid (6,1);
	
	\draw[red] (0.5,0.5) node {$p_1$};
	\draw[blue] (1.5,0.5) node {$p_2$};
	\draw[blue] (2.5,0.5) node {$p_2$};
	
	\draw[red] (3.5,0.5) node {$p_1$};
	\draw[blue] (4.5,0.5) node {$p_2$};
	\draw[font=\small] (5.25,0.8) node {$1-$};
	\draw[font=\small] (5.53,0.5) node {$2p_1-$};
	\draw[font=\small] (5.75,0.2) node {$3p_2$};
	
\end{tikzpicture}

\par\end{centering}
\caption{Input vector with $l=6$ and $d=2$. In this example, $t_1=2$, $t_2=3$ and $t_3=1$. We then have $c_1=\frac{1}{2}$ and $c_2=\frac{1}{3}\left(1-2p_1 \right)$.}
\label{fig:input_vector_t_and_c}
\end{figure}

Next we define the following ``coloring'' function:
\begin{equation} \label{eq:coloring}
		\mathcal{C}(a_j,x_j) =
		\left\{
			\begin{array}{ll}
				k  &  \mathrm{Q}_{A_1|X_1}(a_j|x_j)= p_k\\
				d+1 & \mbox{otherwise}
			\end{array}
		\right.
	\end{equation}
For every pair of strings $(a,x)$, where $a\in \{0,1,\dotsc,l-1\}^n$ and $x\in \{0,1,\dotsc,m-1\}^n$, we denote by $N_i$ the number of indices $j\in[n]$ for which $\mathcal{C}(a_j,x_j) = i$. Using this definition we have $N_{d+1}=n-\sum_{j=1}^d N_j \;$\footnote{Note that the $N_i$'s are functions of the strings $a$ and $x$.}.

Using the notation above, we can now define our measure to be\footnote{Remember that the $c_i$'s are functions of other parameters, therefore $c_1\dotsm c_d$ is not a constant and not even symmetric regarding the different parameters.}
\begin{equation*}
	\mathrm{d}Q_{A_1|X_1} \equiv \frac{\mathrm{d}p_1}{c_1} \frac{\mathrm{d}p_2}{c_2} \dotsm \frac{\mathrm{d}p_d}{c_d} 
\end{equation*}
and use it to define our de Finetti-type state. 

\begin{defn} \label{def:de-Finetti_state}
	For any symmetry $\mathcal{S}$ and the matching state $Q_{A_1|X_1}$ as above, the de Finetti state $\tau^{\mathcal{S}}_{A|X}$ is defined by
	\begin{align*}
		& \tau^{\mathcal{S}}_{A|X}(a|x) = \int Q_{A_1|X_1}^{\otimes n}(a|x) \mathrm{d}Q_{A_1|X_1}  \\
		& \equiv \int_0^{c_1} \frac{\mathrm{d}p_1}{c_1} \int_0^{c_2} \frac{\mathrm{d}p_2}{c_2}  \dotsi \int_0^{c_d} \frac{\mathrm{d}p_d}{c_d} \; p_1^{N_1} \cdot p_2^{N_2} \dotsm p_d^{N_d}  \\ 
		&\quad \times \left[ \frac{1}{t_{d+1}} \left( 1- \sum_{j=1}^d t_j p_j \right) \right]^{n-\sum_{j=1}^d N_j} \;.
	\end{align*}
\end{defn} 

\subsection{Lower bounding the de Finetti state}
The following lemma is the analogous of Lemma \ref{lem:chsh-de-finetti-bound} in Appendix \ref{sec:chsh-proof}. 

\begin{lem} \label{lem:de-Finetti_bound}
	The following lower bound on $\tau^{\mathcal{S}}_{A|X}(a|x)$ holds for all $a,x$
	\begin{equation*}
		\tau^{\mathcal{S}}_{A|X}(a|x) \geq \prod_{j=1}^{d+1} \left( \frac{1}{t_j} \right)^{N_j} {n \choose {N_1, \dotsi ,N_d,N_{d+1}}}^{-1} \frac{1}{(n+1)^{d}}
	\end{equation*}
	where ${n \choose {N_1, \dotsi ,N_d,N_{d+1}}}$ is a multinomial coefficient. 
\end{lem}

Before we continue to the proof of Lemma \ref{lem:de-Finetti_bound}, note that although we have chosen a specific ordering of the parameters in the integration in Definition \ref{def:de-Finetti_state}, this ordering does not affect the bound in Lemma \ref{lem:de-Finetti_bound}. Moreover, this bound is optimal in the sense that there is always at least one pair of strings $(a,x)$ for which the equality is reached, and this pair is independent of the chosen order of the integration. 

\begin{proof}

In the proof we use the following formula:
\begin{equation} \label{eq:basic_formula}
	\begin{split}
		\forall c>0 \; \forall n,N\in \mathbb{N}, N\leq n\\
		\int_0^{c}  \frac{\mathrm{d}p}{c} \; p^{N} (c-p)^{n-N}  \ &= c^n \int_0^1 q^{N} (1-q)^{(n-N)}  \mathrm{d}q  \\
		&= c^n \mathrm{B}(n-N+1, N+1)  \\
		&= c^n {n \choose N}^{-1} \frac{1}{n+1} 
	\end{split}
\end{equation}
where B is the Beta function. We also need the following identities:
\begin{align}
	& t_i \cdot c_i  = 1 -  \sum_{j<i}t_jp_j \label{eq:relation_1}\\
	& 1 - \sum_{j<i}t_jp_j  = t_{i-1} ( c_{i-1} - p_{i-1}) \label{eq:relation_2}\\
	& {{n- \sum_{j=1}^i N_j} \choose {N_{i+1}}} \cdot {n \choose {N_1,\dotsc ,N_i, n-\sum_{j=1}^i N_j}} = \nonumber \\ 
	&\qquad{n \choose {N_1, \dotsc , N_{i+1},n-\sum_{j=1}^{i+1} N_j}} \label{eq:multinomial}
\end{align}
\begin{widetext}
We start by proving the following by induction:
\begin{equation} \label{eq:induction}
	\begin{split}
		\int_0^{c_1} \frac{\mathrm{d}p_1}{c_1} \dotsi \int_0^{c_{d}} \frac{\mathrm{d}p_{d}}{c_{d}} \; p_1^{N_1} \dotsm p_{d}^{N_{d}} \left[ t_{d} \left(c_{d}-p_{d} \right) \right]^{n-\sum_{j=1}^{d} N_j} \geq 
		 \prod_{j=1}^{d} \left( \frac{1}{t_j} \right)^{N_j} {n \choose {N_1, \dotsi ,n - \sum_{j=1}^{d} N_j}}^{-1} \frac{1}{(n+1)^{d}}
	\end{split}
\end{equation}

\emph{Base case, $d=1$:} 
\begin{align*}
	\int_0^{c_1}  \frac{\mathrm{d}p_1}{c_1} \; p^{N_1} \left[t_1(c_1-p_1)\right]^{n-N_1} = \left(\frac{1}{t_1}\right)^{N_1} {n \choose N_1}^{-1} \frac{1}{n+1}
\end{align*}
This follows from Equation \eqref{eq:basic_formula} while noting that for the first index we have $c_1=\frac{1}{t_1}$ by definition.

\emph{Induction hypothesis for d-1:}
\begin{equation} \label{eq:induction_hypothesis}
	\begin{split}
		\int_0^{c_1} \frac{\mathrm{d}p_1}{c_1} \dotsi \int_0^{c_{d-1}} \frac{\mathrm{d}p_{d-1}}{c_{d-1}} \; p_1^{N_1} \dotsm p_{d-1}^{N_{d-1}} \left[ t_{d-1} \left(c_{d-1}-p_{d-1} \right) \right]^{n-\sum_{j=1}^{d-1} N_j} \geq \\
		 \prod_{j=1}^{d-1} \left( \frac{1}{t_j} \right)^{N_j} {n \choose {N_1, \dotsi N_{d-1} ,n - \sum_{j=1}^{d-1} N_j}}^{-1} \frac{1}{(n+1)^{d-1}}
	\end{split}
\end{equation}

\emph{Inductive step:}

\begin{align}
	& \int_0^{c_1} \frac{\mathrm{d}p_1}{c_1} \dotsi \int_0^{c_{d}} \frac{\mathrm{d}p_{d}}{c_{d}} \; p_1^{N_1} \dotsm p_{d}^{N_{d}} \left[ t_{d} \left(c_{d}-p_{d} \right) \right]^{n-\sum_{j=1}^{d} N_j} = \nonumber \\
	& \int_0^{c_1} \frac{\mathrm{d}p_1}{c_1} \dotsi \int_0^{c_{d-1}} \frac{\mathrm{d}p_{d-1}}{c_{d-1}} \; p_1^{N_1} \dotsm p_{d-1}^{N_{d-1}} \int_0^{c_d} \frac{\mathrm{d}p_d}{c_d} \; p_d^{N_d} \left[ t_d \left(c_d-p_d \right) \right]^{n-\sum_{j=1}^{d-1} N_j -N_d} = \label{step_1}\\
	\begin{split}
		& \int_0^{c_1} \frac{\mathrm{d}p_1}{c_1} \dotsi \int_0^{c_{d-1}} \frac{\mathrm{d}p_{d-1}}{c_{d-1}} \; p_1^{N_1} \dotsm p_{d-1}^{N_{d-1}}  \\
		& \quad \qquad \qquad \times t_d^{n- N_d - \sum_{j=1}^{d-1} N_j} c_d^{n-\sum_{j=1}^{d-1} N_j } { n-\sum_{j=1}^{d-1} N_j  \choose N_d}^{-1} \frac{1}{n-\sum_{j=1}^{d-1} N_j +1} =  
	\end{split} \label{step_2} \\
	\begin{split}
		& \left( \frac{1}{t_{d}} \right)^{N_{d}} { n-\sum_{j=1}^{d-1} N_j  \choose N_d}^{-1} \frac{1}{n-\sum_{j=1}^{d-1} N_j + 1}   \\
		& \quad \qquad \qquad \times \int_0^{c_1} \frac{\mathrm{d}p_1}{c_1} \dotsi \int_0^{c_{d-1}} \frac{\mathrm{d}p_{d-1}}{c_{d-1}} \; p_1^{N_1} \dotsm p_{d-1}^{N_{d-1}} \left( 1- \sum_{j=1}^{d-1} t_j p_j \right)^{n-\sum_{j=1}^{d-1} N_j} = 
	\end{split}  \label{step_3} 
\end{align}

\begin{align}
	\begin{split}
		& \left( \frac{1}{t_{d}} \right)^{N_{d}} { n-\sum_{j=1}^{d-1} N_j  \choose N_d}^{-1} \frac{1}{n-\sum_{j=1}^{d-1} N_j + 1}   \\
		&  \qquad \quad \times \int_0^{c_1} \frac{\mathrm{d}p_1}{c_1} \dotsi \int_0^{c_{d-1}} \frac{\mathrm{d}p_{d-1}}{c_{d-1}} \; p_1^{N_1} \dotsm p_{d-1}^{N_{d-1}} \left[ t_{d-1} \left(c_{d-1}-p_{d-1} \right) \right]^{n-\sum_{j=1}^{d-1} N_j} \geq 
	\end{split} \label{step_4} \\
	\begin{split}
		& \left( \frac{1}{t_{d}} \right)^{N_{d}} { n-\sum_{j=1}^{d-1} N_j  \choose N_d}^{-1} \frac{1}{n-\sum_{j=1}^{d-1} N_j + 1}   \\
		& \qquad \qquad \qquad \qquad \times \prod_{j=1}^{d-1} \left( \frac{1}{t_j} \right)^{N_j} {n \choose {N_1, \dotsi N_{d-1} ,n - \sum_{j=1}^{d-1} N_j}}^{-1} \frac{1}{(n+1)^{d-1}} \geq 
	\end{split} \label{step_5}\\
	& \prod_{j=1}^{d} \left( \frac{1}{t_j} \right)^{N_j} {n \choose {N_1, \dotsi N_d}}^{-1} \frac{1}{(n+1)^d} 
\end{align}

where we used Equation \eqref{eq:basic_formula} to get from \eqref{step_1} to \eqref{step_2}, Equation \eqref{eq:relation_1} to get from \eqref{step_2} to \eqref{step_3}, Equation \eqref{eq:relation_2} to get from \eqref{step_3} to \eqref{step_4}, the induction hypothesis \eqref{eq:induction_hypothesis} to get from \eqref{step_4} to \eqref{step_5} and Equation \eqref{eq:multinomial} in the last line.

Next, we prove the lemma by using Equation \eqref{eq:induction}:
\begin{align}
	& \tau^{\mathcal{S}}_{A|X}(a|x) = \nonumber \\
	& \int_0^{c_1} \frac{\mathrm{d}p_1}{c_1} \dotsi \int_0^{c_d} \frac{\mathrm{d}p_d}{c_d} \; p_1^{N_1} \dotsm p_d^{N_d} \cdot \left[ \frac{1}{t_{d+1}} \left( 1- \sum_{j=1}^d t_j p_j \right) \right]^{n-\sum_{j=1}^d N_j} = \nonumber \\
	& \left( \frac{1}{t_{d+1}} \right)^{N_{d+1}} \int_0^{c_1} \frac{\mathrm{d}p_1}{c_1} \dotsi \int_0^{c_d} \frac{\mathrm{d}p_d}{c_d} \; \; p_1^{N_1} \dotsm p_d^{N_d} \left( 1- \sum_{j=1}^d t_j p_j \right)^{n-\sum_{j=1}^d N_j} = \label{step_6}\\
	& \left( \frac{1}{t_{d+1}} \right)^{N_{d+1}} \int_0^{c_1} \frac{\mathrm{d}p_1}{c_1} \dotsi \int_0^{c_{d}} \frac{\mathrm{d}p_{d}}{c_{d}} \; p_1^{N_1} \dotsm p_{d}^{N_{d}} \left[ t_{d} \left(c_{d}-p_{d} \right) \right]^{n-\sum_{j=1}^{d} N_j} \geq \label{step_7} \\
	& \left( \frac{1}{t_{d+1}} \right)^{N_{d+1}} \prod_{j=1}^{d} \left( \frac{1}{t_j} \right)^{N_j} {n \choose {N_1, \dotsi N_d}}^{-1} \frac{1}{(n+1)^d} = \label{step_8} \\
	& \qquad \qquad \qquad \qquad \qquad\qquad \qquad \qquad \quad \prod_{j=1}^{d+1} \left( \frac{1}{t_j} \right)^{N_j} {n \choose {N_1, \dotsi N_d}}^{-1} \frac{1}{(n+1)^d} \nonumber
\end{align}
where we used Equation \eqref{eq:relation_2} to get from \eqref{step_6} to \eqref{step_7} and  Equation \eqref{eq:induction} to get from \eqref{step_7} to \eqref{step_8}. \qedhere
\end{widetext}
\end{proof}

\subsection{Upper bounding a permutation invariant state $\mathrm{P}_{A|X}$ with symmetry $\mathcal{S}$}

The following lemma gives us an upper bound on any permutation invariant state $\mathrm{P}_{A|X}(a|x)$ with the symmetry $\mathcal{S}$. This lemma is the analogous of Lemma \ref{lem:upper-bound-permutation} in Appendix \ref{sec:chsh-proof}. 
\begin{lem}\label{lem:symmetry-bound-general}
	For every permutation invariant state $\mathrm{P}_{A|X}(a|x)$ with symmetry $\mathcal{S}$ we have
	\[
		\forall a,x \quad \mathrm{P}_{A|X}(a|x) \leq \prod_{j=1}^{d+1} \left( \frac{1}{t_j} \right)^{N_j} {n \choose {N_1, \dotsi ,N_d,N_{d+1}}}^{-1}  \;.
	\]
\end{lem}

The idea behind the proof is identical to the idea behind the proof of Lemma \ref{lem:upper-bound-permutation}. We bound the value of a specific entry $\mathrm{P}_{A|X}(a|x)$ by counting how many entries $\mathrm{P}_{A|X}(\tilde{a}|x)$ in the same input vector must have the same value as  $\mathrm{P}_{A|X}(a|x)$ due to the symmetry of $\mathrm{P}_{A|X}$. Since the sum of any input vector is 1 this will give us a bound on $\mathrm{P}_{A|X}(a|x)$. 

\begin{proof}
For our counting arguments we use here the same notation of the coloring function $\mathcal{C}$ given in Equation \eqref{eq:coloring} and the definition of $N_i$ thereafter. That is, for any $a\in \{0,1,\dotsc , l-1\}^n$ and $x\in \{0,1,\dotsc , m-1\}^n$, we denote by $N_i$ the number of indices $j\in[n]$ for which $\mathcal{C}(a_j,x_j) = i$. We can imagine this as placing a colored ball for each pair $(a_j,x_j)$ as in Figure \ref{fig:balls_general}. With this picture in mind, the symmetry $\mathcal{S}$ actually says that by changing two balls of the same color we do not change the value according to the probability distribution $\mathrm{P}_{A|X}$. 
Let 
\[
	\mathcal{N}(a,x) = \left| \{ \tilde{a}\in\{0,1,...,l-1\}^n | \mathrm{P}_{A|X}(\tilde{a}|x) =\mathrm{P} _{A|X}(a|x) \} \right| \;.
\] 
In how many ways can we change $a$ while not changing the value of the entry according to $\mathrm{P}_{A|X}$?  
We now prove 
\[
	\mathcal{N}(a,x) \geq \prod_{j=1}^{d+1} t_j^{N_j} {n \choose {N_1, \dotsi ,N_d,N_{d+1}}} \;.
\]
As in the proof of Lemma \ref{lem:upper-bound-permutation} in Appendix \ref{sec:chsh-proof}, we have two different ways of changing $a$ to $\tilde{a}$: with and without changing the color sequence of the balls.

\begin{figure}
\begin{centering}
	\begin{tikzpicture}[dot/.style={circle,inner sep =2pt}]
 	
	\draw (0,1) node {$a:$};
	\draw (0,0.5) node {$x:$};
	\draw[font=\scriptsize] (0, -0.05) node {$j:$};

	\node[dot, fill=red] at (0.75,1.5) {};
	\draw (0.75,1) node {$3$};
	\draw (0.75,0.5) node {$0$};
	\draw[font=\scriptsize] (0.75, -0.05) node {$1$};
		
	\node[dot, fill=blue] at (1.15,1.5) {};
	\draw (1.15,1) node {$1$};
	\draw (1.15,0.5) node {$0$};
	\draw[font=\scriptsize] (1.15, -0.05) node {$2$};
	
	\node[dot, fill=black] at (1.55,1.5) {};
	\draw (1.55,1) node {$0$};
	\draw (1.55,0.5) node {$1$};
	\draw[font=\scriptsize] (1.55, -0.05) node {$3$};

	\node[dot, fill=red] at (1.95,1.5) {};
	\draw (1.95,1) node {$2$};
	\draw (1.95,0.5) node {$1$};
	\draw[font=\scriptsize] (1.95, -0.05) node {$4$};
	
	\node[dot, fill=red] at (2.35,1.5) {};
	\draw (2.35,1) node {$1$};
	\draw (2.35,0.5) node {$1$};
	\draw[font=\scriptsize] (2.35, -0.05) node {$5$};
	
	\node[dot, fill=red] at (2.75,1.5) {};
	\draw (2.75,1) node {$0$};
	\draw (2.75,0.5) node {$0$};
	\draw[font=\scriptsize] (2.75, -0.05) node {$6$};
	
	\node[dot, fill=black] at (3.15,1.5) {};
	\draw (3.15,1) node {$2$};
	\draw (3.15,0.5) node {$0$};
	\draw[font=\scriptsize] (3.15, -0.05) node {$7$};

	\node[dot, fill=blue] at (3.55,1.5) {};
	\draw (3.55,1) node {$3$};
	\draw (3.55,0.5) node {$1$};
	\draw[font=\scriptsize] (3.55, -0.05) node {$8$};
	
	\begin{scope}[shift={(-1.75,2.5)}]
		\draw[step=1] (0,0) grid (8,1);
		\draw[thick] (4,-0.25) -- (4,1.25);
		
		\draw[red] (0.5,0.5) node {$p_1$};
		\draw[blue] (1.5,0.5) node {$p_2$};
		\draw[font=\small] (2.25,0.75) node {$1-$};
		\draw[font=\small] (2.55,0.5) node {$2p_1-$};
		\draw[font=\small] (2.75,0.25) node {$p_2$};
		\draw[red] (3.5,0.5) node {$p_1$};
		
		\draw[font=\small] (4.25,0.75) node {$1-$};
		\draw[font=\small] (4.55,0.5) node {$2p_1-$};
		\draw[font=\small] (4.75,0.25) node {$p_2$};
		\draw[red] (5.5,0.5) node {$p_1$};	
		\draw[red] (6.5,0.5) node {$p_1$};	
		\draw[blue] (7.5,0.5) node {$p_2$};	
	\end{scope}
	
	\end{tikzpicture}

\par\end{centering}
\caption{The coloring of the pairs $(a_j,x_j)$ according to the structure of $Q_{A|X}$ on the top. Here we have $n=8$, $N_1=4$, $N_2=2$ and $N_3=2$.} \label{fig:balls_general}
\end{figure}

Indeed, the first possible way to change $a$ without changing the value of the entry is to a change a pair $(a_j,x_j)$ to a pair $(\tilde{a_j},x_j)$ of the same color (note that we do not change $x_j$ since we want to stay in the same input vector of $\mathrm{P}_{A|X}$, i.e., not to change the input $x$). In the example of Figure \ref{fig:balls_general} we can change the first pair $(a_1,x_1) = (3,0)$ to $(\tilde{a_1},x_1)=(0,0)$ for example. How many different strings $\tilde{a}$ can we create this way? In each input vector of $Q_{A|X}$ we have $t_j$ entries of the $j$'th color\footnote{We mention the input vectors of $Q_{A|X}$ here just for simplicity. What we really mean is that we have $t_j$ symmetry conditions, but these were ``constructed'' from $Q_{A|X}$ in Definition \ref{defn:symmetry_relation}.} and we can choose a entry with this color for each one of the $N_j$ indices with this color. Therefore, there are exactly $\prod_{j=1}^{d+1} t_j^{N_j} $ different strings $\tilde{a}$ with the same color sequence as $a$ and hence, according to the symmetry $\mathcal{S}$, with the same value $\mathrm{P}_{A|X}(a,x)=\mathrm{P}_{A|X}(\tilde{a},x)$.

Changing $a$ in different ways than the way given above will necessarily change the colors sequence. However, we can still prove, using the permutation invariance of $\mathrm{P}_{A|X}$, that for some specific changes the value of the entry will stay the same. The specific changes that we consider are derived by permutations of the colored balls. 

In order to understand how every different permutations of the balls $\pi$ is realised as a permutation on $x$ and $a$ consider the example drawn in Figure \ref{fig:permutation_general}. On the left side we see a permutation of the balls from Figure \ref{fig:balls_general}. We start by filling up the columns for which there is no change in the color of the ball with the original columns. Then we fill in the blank columns in such a way that each of the original columns appears once. The permutation in the figure for example, is just the permutation of indices $(3,4)$ and $(6,7,8)$. In general, there might be several ways to choose the permutation on  $x$ and $a$, but they are all equivalent for our purpose and therefore we can just choose one. 

\begin{figure*}
\begin{centering}
	\begin{tikzpicture}[dot/.style={circle,inner sep =2pt}]
 	
\begin{scope}
	\draw (0,1) node {$\pi(a):$};
	\draw (0,0.5) node {$\pi(x):$};
	\draw[font=\scriptsize] (0, -0.05) node {$\pi(j):$};

	\node[dot, fill=red] at (0.75,1.5) {};
	\draw (0.75,1) node {$$};
	\draw (0.75,0.5) node {$$};
	\draw[font=\scriptsize] (0.75, -0.05) node {$$};
		
	\node[dot, fill=blue] at (1.15,1.5) {};
	\draw (1.15,1) node {$$};
	\draw (1.15,0.5) node {$$};
	\draw[font=\scriptsize] (1.15, -0.05) node {$$};
	
	\node[dot, fill=red] at (1.55,1.5) {};
	\draw (1.55,1) node {$$};
	\draw (1.55,0.5) node {$$};
	\draw[font=\scriptsize] (1.55, -0.05) node {$$};

	\node[dot, fill=black] at (1.95,1.5) {};
	\draw (1.95,1) node {$$};
	\draw (1.95,0.5) node {$$};
	\draw[font=\scriptsize] (1.95, -0.05) node {$$};
	
	\node[dot, fill=red] at (2.35,1.5) {};
	\draw (2.35,1) node {$$};
	\draw (2.35,0.5) node {$$};
	\draw[font=\scriptsize] (2.35, -0.05) node {$$};
	
	\node[dot, fill=blue] at (2.75,1.5) {};
	\draw (2.75,1) node {$$};
	\draw (2.75,0.5) node {$$};
	\draw[font=\scriptsize] (2.75, -0.05) node {$$};
	
	\node[dot, fill=red] at (3.15,1.5) {};
	\draw (3.15,1) node {$$};
	\draw (3.15,0.5) node {$$};
	\draw[font=\scriptsize] (3.15, -0.05) node {$$};

	\node[dot, fill=black] at (3.55,1.5) {};
	\draw (3.55,1) node {$$};
	\draw (3.55,0.5) node {$$};
	\draw[font=\scriptsize] (3.55, -0.05) node {$$};
\end{scope}

\begin{scope}[shift={(5,0)}]
	\draw (0,1) node {$\pi(a):$};
	\draw (0,0.5) node {$\pi(x):$};
	\draw[font=\scriptsize] (0, -0.05) node {$\pi(j):$};

	\node[dot, fill=red] at (0.75,1.5) {};
	\draw (0.75,1) node {$3$};
	\draw (0.75,0.5) node {$0$};
	\draw[font=\scriptsize] (0.75, -0.05) node {$1$};
		
	\node[dot, fill=blue] at (1.15,1.5) {};
	\draw (1.15,1) node {$1$};
	\draw (1.15,0.5) node {$0$};
	\draw[font=\scriptsize] (1.15, -0.05) node {$2$};
	
	\node[dot, fill=red] at (1.55,1.5) {};
	\draw (1.55,1) node {$$};
	\draw (1.55,0.5) node {$$};
	\draw[font=\scriptsize] (1.55, -0.05) node {$$};

	\node[dot, fill=black] at (1.95,1.5) {};
	\draw (1.95,1) node {$$};
	\draw (1.95,0.5) node {$$};
	\draw[font=\scriptsize] (1.95, -0.05) node {$$};
	
	\node[dot, fill=red] at (2.35,1.5) {};
	\draw (2.35,1) node {$1$};
	\draw (2.35,0.5) node {$1$};
	\draw[font=\scriptsize] (2.35, -0.05) node {$5$};
	
	\node[dot, fill=blue] at (2.75,1.5) {};
	\draw (2.75,1) node {$$};
	\draw (2.75,0.5) node {$$};
	\draw[font=\scriptsize] (2.75, -0.05) node {$$};
	
	\node[dot, fill=red] at (3.15,1.5) {};
	\draw (3.15,1) node {$$};
	\draw (3.15,0.5) node {$$};
	\draw[font=\scriptsize] (3.15, -0.05) node {$$};

	\node[dot, fill=black] at (3.55,1.5) {};
	\draw (3.55,1) node {$$};
	\draw (3.55,0.5) node {$$};
	\draw[font=\scriptsize] (3.55, -0.05) node {$$};
\end{scope}

\begin{scope}[shift={(10,0)}]
	\draw (0,1) node {$\pi(a):$};
	\draw (0,0.5) node {$\pi(x):$};
	\draw[font=\scriptsize] (0, -0.05) node {$\pi(j):$};

	\node[dot, fill=red] at (0.75,1.5) {};
	\draw (0.75,1) node {$3$};
	\draw (0.75,0.5) node {$0$};
	\draw[font=\scriptsize] (0.75, -0.05) node {$1$};
		
	\node[dot, fill=blue] at (1.15,1.5) {};
	\draw (1.15,1) node {$1$};
	\draw (1.15,0.5) node {$0$};
	\draw[font=\scriptsize] (1.15, -0.05) node {$2$};
	
	\node[dot, fill=red] at (1.55,1.5) {};
	\draw (1.55,1) node {$2$};
	\draw (1.55,0.5) node {$1$};
	\draw[font=\scriptsize, orange] (1.55, -0.05) node {$4$};

	\node[dot, fill=black] at (1.95,1.5) {};
	\draw (1.95,1) node {$0$};
	\draw (1.95,0.5) node {$1$};
	\draw[font=\scriptsize, orange] (1.95, -0.05) node {$3$};
	
	\node[dot, fill=red] at (2.35,1.5) {};
	\draw (2.35,1) node {$1$};
	\draw (2.35,0.5) node {$1$};
	\draw[font=\scriptsize] (2.35, -0.05) node {$5$};
	
	\node[dot, fill=blue] at (2.75,1.5) {};
	\draw (2.75,1) node {$3$};
	\draw (2.75,0.5) node {$1$};
	\draw[font=\scriptsize, purple] (2.75, -0.05) node {$8$};
	
	\node[dot, fill=red] at (3.15,1.5) {};
	\draw (3.15,1) node {$0$};
	\draw (3.15,0.5) node {$0$};
	\draw[font=\scriptsize, purple] (3.15, -0.05) node {$6$};

	\node[dot, fill=black] at (3.55,1.5) {};
	\draw (3.55,1) node {$2$};
	\draw (3.55,0.5) node {$0$};
	\draw[font=\scriptsize, purple] (3.55, -0.05) node {$7$};
\end{scope}

\end{tikzpicture}

\par\end{centering}
\caption{The permutation $\pi$.} \label{fig:permutation_general}
\end{figure*}

The important thing to note is that between every two different permutations of the balls we always have at least one index in which we have a different colored ball in the end. That is, we can write that for every $\pi,\pi'\neq\pi$, there exists $j\in[n]$ such that $\mathcal{C}(\pi(a)_j,\pi(x)_j) \neq \mathcal{C}(\pi'(a)_j,\pi'(x)_j)$. We use this to construct from every permutation $\pi$ a different string $\tilde{a}_\pi$ for which $\mathrm{P}_{A|X}(a|x)=\mathrm{P}_{A|X}(\tilde{a}_\pi|x)$, as follows. For any index $j\in[n]$ that $\pi$ permutes we change $a_j$ to (some) $\tilde{a_j}_\pi$ such that 
\begin{equation} \label{eq:color_change}
	\mathcal{C}(\tilde{a_j}_\pi,x_j)=\mathcal{C}(\pi(a)_j,\pi(x)_j) \;.
\end{equation}
This is always possible since the input vectors of $Q_{A|X}$ are permutations of one another, i.e., if $\mathcal{C}(\pi(a)_j,\pi(x)_j)=k$ then there must be some $a'_j$ for which $\mathcal{C}(a'_j,x_j)=k$ \footnote{Again, as in the previous footnote, what we really mean is that this holds according to the symmetry $\mathcal{S}.$}. 

We are now left to show that $\mathrm{P}_{A|X}(a,x)=\mathrm{P}_{A|X}(\tilde{a}_\pi,x)$. Since $\mathrm{P}_{A|X}$ is permutation invariant, we have 
\[
	\mathrm{P}_{A|X}(a|x) = \mathrm{P}_{A|X}\left( \pi(a)|\pi(x) \right) 
\]
and from the symmetry $\mathcal{S}$ (recall Definition \ref{defn:symmetry_relation}) and Equation \eqref{eq:color_change} we get
\[
	\mathrm{P}_{A|X}\left( \pi(a)|\pi(x) \right)  = \mathrm{P}_{A|X}\left( \tilde{a}_{\pi}|x \right) \;.
\]
Combining these two equations together we get $\mathrm{P}_{A|X}(a,x)=\mathrm{P}_{A|X}(\tilde{a}_\pi,x)$ as desired. 

Since for every two different permutations of the balls we always have at least one index in which we have a different colored ball in the end, we get different $\tilde{a}_{\pi}$'s from different permutations $\pi$. There are exactly ${n \choose {N_1, \dotsi ,N_d,N_{d+1}}}$ different permutations of the balls, and therefore the same number of different  $\tilde{a}_{\pi}$ when proceeding this way. 

We can now answer our original question and bound $\mathcal{N}(a,x)$. We can combine both of the ways given above to change $a$ without changing the value of the entry according to $\mathrm{P}_{A|X}$. This implies that in total, there are at least $\prod_{j=1}^{d+1} t_j^{N_j} \times {n \choose {N_1, \dotsi ,N_d,N_{d+1}}}$ different ways of changing $a$ and we can conclude that
\begin{equation} \label{eq:N_bound}
	\mathcal{N}(a,x) \geq \prod_{j=1}^{d+1} t_j^{N_j} {n \choose {N_1, \dotsi ,N_d,N_{d+1}}} \;.
\end{equation}

Since for all $x$ $\sum_a \mathrm{P}_{A|X}(a|x) = 1$, we get from Equation \eqref{eq:N_bound} the following bound on the entry value $\mathrm{P}_{A|X}(a|x)$:
\begin{equation*}
	\mathrm{P}_{A|X}(a|x) \leq \prod_{j=1}^{d+1} \left( \frac{1}{t_j} \right)^{N_j} {n \choose {N_1, \dotsi ,N_d,N_{d+1}}}^{-1}  \;. \qedhere
\end{equation*}

\end{proof}

Combining Lemma \ref{lem:de-Finetti_bound} and \ref{lem:symmetry-bound-general} we get Theorem~\ref{thm:post-selection}.

\subsection{Deriveing the corollaries from the general theorem}\label{sec:cor_derivation}

As mentioned before, for every symmetry $\mathcal{S}$ for which a state $Q_{A_1|X_1}$ can be construct such that the condition in Definition \ref{defn:symmetry_relation} holds our proof can be applied. In order to derive the corollaries we just need to describe the type of symmetry that we consider and the relevant $Q_{A_1|X_1}$ that we use to construct the de Finetti state. 

Consider for example Corollary \ref{cor:conditional}, where the states $\mathrm{P}_{A|X}$ have no special symmetry. We can therefore derive the corollary from Theorem \ref{thm:post-selection} by choosing $Q_{A_1|X_1}$ without any internal symmetry (see, e.g, Figure~\ref{fig:general-state-no-symmetry}). In this case we have $d=m(l-1)$ degrees of freedom, hence we get Corollary \ref{cor:conditional}. 
For deriving Corollary \ref{cor:chsh-post-selection} we use the state $Q_{A_1B_1|X_1Y_1}$ given in Figure~\ref{fig:CHSH_symmetry}. For this state we have $d=1$ and Corollary \ref{cor:chsh-post-selection} follows.

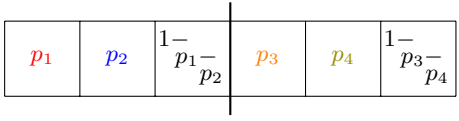
\begin{figure}
\begin{centering}
\begin{tikzpicture}

	\draw[step=1] (0,0) grid (6,1);
	\draw[thick] (3,-0.25) -- (3,1.25);
	
	\draw[red] (0.5,0.5) node {$p_1$};
	\draw[blue] (1.5,0.5) node {$p_2$};
	\draw[font=\small] (2.25,0.75) node {$1-$};
	\draw[font=\small] (2.55,0.5) node {$p_1-$};
	\draw[font=\small] (2.75,0.25) node {$p_2$};
	
	\draw[orange] (3.5,0.5) node {$p_3$};	
	\draw[olive] (4.5,0.5) node {$p_4$};	
	\draw[font=\small] (5.25,0.75) node {$1-$};
	\draw[font=\small] (5.55,0.5) node {$p_3-$};
	\draw[font=\small] (5.75,0.25) node {$p_4$};

\end{tikzpicture}

\par\end{centering}
\caption{A ststem $Q_{A_1|X_1}$ without any internal symmetry, with $m=2$ and $l=3$.}
\label{fig:general-state-no-symmetry}
\end{figure}

\section{Applications to cryptography} \label{sec:diamond-proofs}

As explained in the main text, the main motivation for de Finetti reductions is that they allow us to simplify the analysis of information theoretical processes. In this section we explain how to use the de Finetti reductions in applications. For this, we will use an alternative formulation of the de Finetti reductions given above. We therefore first explain in Section \ref{sec:post_selection} this alternative formulation and show how to derive it from the de Finetti reductions.

In Section \ref{sec:test_app} we consider the simplified application given in the main text -- bounding the failure probability of a permutation invariant test that interacts with the input state. We finish in Section \ref{sec:diamond_app} by considering the diamond norm, which is the relevant measure to bound in the context of cryptography. We define the diamond norm for conditional probability distributions and show how to bound it using the de Finetti reductions. 

\subsection{Post selecting a permutation invariant state from a de Finetti state} \label{sec:post_selection}

Before starting, we will need the following definition. 
\begin{defn} 
	An extension\footnote{In quantum physics, a purification is a special case of an extension.} of a state $\mathrm{P}_{A|X}$ is a state $\mathrm{P}_{AC|XZ}$ such that $\forall z\in Z, \; \mathrm{P}_{A|X}(a|x) = \sum_c \mathrm{P}_{AC|XZ}(ac|xz)$.
	We say that an extension $\mathrm{P}_{AC|XZ}$  is non-signalling if the second marginal, $ \mathrm{P}_{C|Z}$ is also properly defined, i.e., it does not depend on $x$. 
\end{defn}
For simplicity (and since it is the relevant scenario for cryptography) we consider here only non-signalling extensions $\mathrm{P}_{AC|XZ}$ of $\mathrm{P}_{A|X}$. 

We now give an alternative formulation of the de Finetti reductions. We explain and derive this alternative formulation only for Corollary \ref{cor:conditional}, but it can be applied analogously also to the other de Finetti reductions. 

\begin{lem}\label{lem:conditional}
	There exists a de Finetti state $\tau_{A|X}$ where $x \in \{ 0,1, ... ,m-1 \}^n$ and $a \in \{ 0,1, ... ,l-1 \} ^n$ and a non-signalling extension of it to a larger state $\tau_{AC|XZ}$ such that for every permutation invariant state $\mathrm{P}_{A|X}$ there exists a measurement $z$ and an outcome of this measurement $c_z$ for which 
	\[
		\forall a,x \quad \tau_{AC|XZ} (a,c_z|x,z) = \frac{1}{(n+1)^{m(l-1)}}\mathrm{P}_{A|X} (a|x) \;.
	\]
\end{lem}

This lemma states that there exists a de Finetti state $\tau_{A|X}$ and a non-signalling extension of it $\tau_{AC|XZ}$ such that any permutation invariant state $\mathrm{P}_{A|X}$ can be post selected from it with probability $\geq \frac{1}{(n+1)^{m(l-1)}}$. When we say that  $\mathrm{P}_{A|X}$ can be post selected we mean that there exists an input $z$ to $\tau_{AC|XZ}$ and an output of this measurement $c_z$ such that with probability $\tau_{C|Z}(c_z|z)\geq\frac{1}{(n+1)^{m(l-1)}}$ the post-measurement state is $\mathrm{P}_{A|X}$ (see Figure \ref{fig:postselectionl}). Note that we consider a specific extension $\tau_{AC|XZ}$ of the state $\tau_{A|X}$, and by choosing different inputs $z$ we can post select different states $\mathrm{P}_{A|X}$.

\begin{figure}
	\begin{centering}
	\begin{tikzpicture}
			
		\node at (1,0.5) {$\tau_{AC|XZ}$};
		\node at (0.5,1.7) {x};
		\draw[->] (0.5,1.5) -- (0.5,1);
		\node[red] at (1.5,1.7) {z};
		\draw[->,red] (1.5,1.5) -- (1.5,1);
		\draw  (0,0) rectangle (2,1);
		\node at (0.5,-0.7) {a};
		\draw[->] (0.5,0) -- (0.5,-0.5);
		\node[red] at (1.5,-0.7) {$c_z$};
		\draw[->,red] (1.5,0) -- (1.5,-0.5);
		
		\node at (2.5,0.5) {=};
		
		\node at (3.5,0.5) {$\mathrm{P}_{A|X}$};
		\node at (3.5,1.7) {x};
		\draw[->] (3.5,1.5) -- (3.5,1);
		\draw (3,0) rectangle (4,1);
		\node at (3.5,-0.7) {a};
		\draw[->] (3.5,0) -- (3.5,-0.5);
		  	
	\end{tikzpicture}

	\end{centering}
\caption{Post selecting state $\mathrm{P}_{A|X}$ from an extension of $\tau_{A|X}$. If the outcome of the measurement $z$ is $c_z$ then the post measurement state is $\mathrm{P}_{A|X}$.}
\label{fig:postselectionl}
\end{figure}
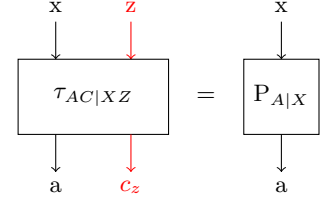

It is easy to see how to derive Lemma \ref{lem:conditional} from Corollary~\ref{cor:conditional} by using the formalism introduced in \cite{hanggi2009quantum,hanggi2010device} of partitions of a conditional probability distribution. We repeat here the relevant statements. 
\begin{defn}
	A partition of a state $\mathrm{P}_{A|X}$ is a family of pairs $\left\{ \left( p_c,\mathrm{P}^c_{A|X} \right) \right\}_c$ where $p_c\geq 0$, $\sum_c p_c = 1$ and the states $\mathrm{P}^c_{A|X}$  are such that
	\[
		\mathrm{P}_{A|X} = \sum_c p_c \cdot \mathrm{P}^c_{A|X} \;.
	\]
\end{defn}
\begin{lem} [Lemma 9 in \cite{hanggi2009quantum}] \label{lem:extension_condition}
	Given a state $\mathrm{P}_{A|X}$, there exists a partition with element $\left( p_c,\mathrm{P}^c_{A|X} \right)$ if and only if 
	\[
		\forall a,x \quad p_c\cdot \mathrm{P}^c_{A|X}(a|x) \leq \mathrm{P}_{A|X} \;.
	\]
\end{lem}

\begin{lem}[Lemma 3.2 in \cite{hanggi2010device}] \label{lem:extension}
	Given a state $\mathrm{P}_{A|X}$ let $Z$ be the set of all partitions $\left\{ \left( p_{c_z},\mathrm{P}^{c_z}_{A|X} \right) \right\}_{c_z}$ of $\mathrm{P}_{A|X}$. There exists an extension state $\mathrm{P}_{AC|XZ}$  of $\mathrm{P}_{A|X}$ and an input $z$ to it such that  
	\[
		\forall a,x \quad \mathrm{P}_{AC|XZ}(a,c_z|x,z)=p_{c_z} \cdot \mathrm{P}^{c_z}_{A|X}(a|x) \;.
	\]
	Moreover, the state $\mathrm{P}_{AC|XZ}$ does not allow signalling between the $A/X$ and the $C/Z$ interfaces\footnote{In the usuall cryptographic setting this means a non-signalling condition between Alice and Eve.}.
\end{lem}

Using the lemmas above and Corollary \ref{cor:conditional} we can now prove Lemma \ref{lem:conditional}. 
\begin{proof}
	The above lemmas together with Corollary \ref{cor:conditional} imply that in our case for any permutation invariant state $\mathrm{P}_{A|X}$, $\left( \frac{1}{(n+1)^{m(l-1)}},\mathrm{P}_{A|X} \right) $ is an element of a partition of $\tau_{A|X}$. Moreover, there exists a state $\tau_{AC|XZ}$ and an input $z$ such that with probability $\frac{1}{(n+1)^{m(l-1)}}$ the post-measurement state is $\mathrm{P}_{A|X}$:
	\[
			\forall a,x \quad \tau_{AC|XZ}(a,c_z|x,z) = \frac{1}{(n+1)^{m(l-1)}}\mathrm{P}_{A|X} \;. \qedhere
	\]
\end{proof}

\subsection{Simplified appliaction} \label{sec:test_app}

Let $\mathcal{T}$ be a test which interacts with a state $\mathrm{P}_{A|X}$ and outputs ``success'' or ``fail'' with some probabilities.  We denote by $\mathrm{Pr}_{\text{fail}}(\mathrm{P}_{A|X})$ the probability that $\mathcal{T}$ outputs ``fail'' after interacting with $\mathrm{P}_{A|X}$. We consider permutation invariant tests, as in Definition \ref{def:permutation-invariant-test} and prove Lemma \ref{lem:test-bound} using the de Finetti reduction of Corollary \ref{cor:conditional}.

\begin{proof}
We follow here a similar proof given in \cite{renner2010simplifying} for the quantum post selection theorem \cite{christandl2009postselection}.

 First, since the test $\mathcal{T}$ is permutation invariant it is sufficient to consider only permutation invariant states. To see this recall that for any state  $\mathrm{P}_{A|X}$  and permutation $\pi$ we have $\mathrm{Pr}_{\text{fail}}(\mathrm{P}_{A|X}) = \mathrm{Pr}_{\text{fail}}(\mathrm{P}_{A|X}\circ\pi)$ according to Definition  \ref{def:permutation-invariant-test}. Therefore we also have by linearity 
\begin{align} \label{eq:only-permutations}
	\mathrm{Pr}_{\text{fail}}(\mathrm{P}_{A|X}) &= \frac{1}{n!} \sum_\pi \mathrm{Pr}_{\text{fail}}(\mathrm{P}_{A|X}\circ\pi) \nonumber \\
	&=\mathrm{Pr}_{\text{fail}} \left(  \frac{1}{n!} \sum_\pi \mathrm{P}_{A|X}\circ\pi \right) \;.
\end{align}
The state $\frac{1}{n!} \sum_\pi \mathrm{P}_{A|X}\circ\pi$ is permutation invariant and therefore without loss of generality we can consider only permutation invariant states. 

Next we define the following probabilities. Let $\mathrm{Pr}_{\text{fail}\land c_z}(\tau_{AC|XZ})$ be the probability that the second part of the state, $\tau_{C|Z}$ is measured with $z$ and the output is $c_z$ and that the first part of the state, $\tau_{A|X}$ fails the test $\mathcal{T}$ at the same time. That is, 
\[
	\mathrm{Pr}_{\text{fail}\land c_z}(\tau_{AC|XZ}) = \mathrm{Pr}_{\text{fail}}(\tau_{A|X}) \cdot \tau_{C|Z}(c_z|z) \;.
\]
In a similar way we define $\mathrm{Pr}_{\text{fail}|c_z}(\tau_{AC|XZ})$ to be the probability the $\tau_{A|X}$ fails the test $\mathcal{T}$ given that $c_z$ is the outcome measurement of $\tau_{C|Z}$. According to probability theory we have 
\[
	\mathrm{Pr}_{\text{fail}|c_z}(\tau_{AC|XZ}) = \frac{\mathrm{Pr}_{\text{fail}\land c_z}(\tau_{AC|XZ})}{\tau_{C|Z}(c_z|z)} \leq \frac{\mathrm{Pr}_{\text{fail}}(\tau_{A|X})}{\tau_{C|Z}(c_z|z)}
\]
since it is always true that $\mathrm{Pr}_{\text{fail}\land c_z}(\tau_{AC|XZ}) \leq \mathrm{Pr}_{\text{fail}}(\tau_{A|X})$.

Lemma \ref{lem:conditional} implies that $\tau_{C|Z}(c_z|z)\geq\frac{1}{(n+1)^{m(l-1)}}$ and that $\mathrm{Pr}_{\text{fail}|c_z}(\tau_{AC|XZ})=\mathrm{Pr}_{\text{fail}}(\mathrm{P}_{A|X})$ (given that the outcome measurement was $c_z$, the post measurement state is $\mathrm{P}_{A|X}$). All together we get $\mathrm{Pr}_{\text{fail}}(\mathrm{P}_{A|X}) \leq (n+1)^{m(l-1)} \mathrm{Pr}_{\text{fail}}(\tau_{A|X})$ as required. 
\end{proof}

As explained in the main text, the importance of the de Finetti reductions is already obvious from this simplified example. Lemma \ref{lem:test-bound} tells us that instead of proving an upper bound on the failure probability of the test $\mathcal{T}$ for states $\mathrm{P}_{A|X}$, it is sufficient to prove it for the de Finetti state $\tau_{A|X}$ and ``pay'' for it with the additional polynomial factor of $(n+1)^{m(l-1)}$.

\subsection{Bounding the diamond norm for conditional probability distributions} \label{sec:diamond_app}

While the notion of a test as discussed above allows for a simple treatment of cases where the failure can be directly defined as an event, it is unfortunately not directly applicable to security proofs for general cryptographic protocols, such as quantum key distribution. In order to prove security one usually needs to establish an upper bound on the distinguishing advantage between the applied protocol and an ideal protocol. Formally we describe the protocols by channels which act on the state and bound the distinguishing advantage between the two channels.

When considering quantum protocols this distinguishing advantage is given by the diamond norm \cite{kitaev1997quantum}. The distance between two channels $\mathcal{E}$ and $\mathcal{F}$ which act on quantum states $\rho_A$ is given by $\| \mathcal{E}-\mathcal{F}\|_{\diamond}=\underset{\rho_{AC}}{\mathrm{max}}\| \left(\mathcal{E}-\mathcal{F}\right)\otimes \mathbb{1} \; \rho_{AC} \|_1$ where $\rho_{AC}$ is a purification of $\rho_A$ and $\| \cdot \|_1$ is the trace distance. Informally the idea is that in order to distinguish two channels we are not only allowed to choose the input state to the channels, $\rho_A$, but also keep to ourselves a purifying state  $\rho_C$. 

Although the definition of the diamond norm includes a maximisation over all states $\rho_{AC}$, using the quantum post selection theorem, it was proven that when considering permutation invariant channels it is sufficient to calculate the distance for a specific quantum de Finetti state~\cite{christandl2009postselection}. Motivated by this we prove here a similar bound on a distance analogous to the diamond norm for channels which act on conditional probability distributions. 

We consider here channels of the form $\mathcal{E}:\{\mathrm{P}_{A|X}\}\rightarrow \{\mathrm{P}_K\}$ which interact with conditional probability distributions $\mathrm{P}_{A|X}$ and output a classical bit string $k\in\{0,1\}^t$ of some length $t\geq 0$ with some probability $\mathrm{P}_K(k)$. The probability distribution of the output depends on the channel $\mathcal{E}$ itself and is given by the following definition.
\begin{defn}
	The probability that a channel $\mathcal{E}$ will output a string $k\in\{0,1\}^{t}$ when interacting with $\mathrm{P}_{A|X}$ is 
	\[
		\mathrm{E}_K(k) = \sum_x \mathrm{Pr}_{\mathcal{E}}(x) \sum_{a |  \mathcal{E}(a,x)=k} \mathrm{P}_{A|X}(a|x)
	\]
	where $\mathrm{Pr}_{\mathcal{E}}(x)$ is the probability that the channel $\mathcal{E}$ will input $x$ to $\mathrm{P}_{A|X}$ and $\mathcal{E}(a,x)$ is the function according to which the output of the channel is determined. Analogously, 
	\[
		\mathrm{E}_{K|C}(k|c) = \sum_x \mathrm{Pr}_{\mathcal{E}}(x) \sum_{a | \mathcal{E}(a,x)=k} \mathrm{P}_{A|XC}(a|xc) \;.
	\] 
\end{defn}
The connection between the channel and the state is illustrated in Figure \ref{fig:channel}.

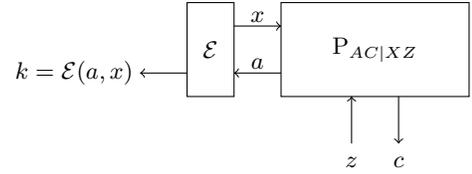
\begin{figure}[b]
	\begin{centering}
	\begin{tikzpicture}[scale=1.25]
	\begin{scope}[shift={(-0.5,-0.5)}]
		\node at (1,0.5) {$\mathrm{P}_{AC|XZ}$};
		\draw  (0,0) rectangle (2,1);
		
		\node at (-0.25,0.85) {$x$};
		\draw[->] (-0.5,0.75) -- (0,0.75);
		\node at (-0.25,0.35) {$a$};
		\draw[<-] (-0.5,0.25) -- (0,0.25);
	
		\node at (0.75,-0.7) {$z$};
		\draw[<-] (0.75,0) -- (0.75,-0.5);
		\node at (1.25,-0.7) {$c$};
		\draw[->] (1.25,0) -- (1.25,-0.5);
		
		\draw (-1,0) rectangle (-0.5,1);
		\node at (-0.75,0.5) {$\mathcal{E}$};
		\draw[->] (-1,0.25) -- (-1.5,0.25);
		\node at (-2.2,0.25) {$k=\mathcal{E}(a,x)$};
	\end{scope}
				  	
	\end{tikzpicture}
	\end{centering}
\caption{The channel $\mathcal{E}\otimes \mathbb{1}$ acts on an extension $\mathrm{P}_{AC|XZ}$ of  $\mathrm{P}_{A|X}$ and outputs a classical string $k\in\{0,1\}^t$ according to the probability $\mathrm{E}_K(k)$.}
\label{fig:channel}
\end{figure}

\begin{defn}\label{def:diamond-norm}
	The distance between two channels $\mathcal{E},\mathcal{F}:\{\mathrm{P}_{A|X}\}\rightarrow\{\mathrm{P}_K\}$ according to the diamond norm is 
	\[
		\| \mathcal{E}-\mathcal{F}\|_{\diamond}=\underset{\mathrm{P}_{AC|XZ}}{\mathrm{max}}\| \left(\mathcal{E}-\mathcal{F}\right)\otimes \mathbb{1}(\mathrm{P}_{AC|XZ} )\|_1 \;,
	\]
	where the maximisation is over all states $\mathrm{P}_{A|X}$ and all possible extensions of them and
	\[
	\begin{split}
		\mathcal{E}\otimes \mathbb{1}(\mathrm{P}_{AC|XZ}) &= \mathcal{E}\otimes \mathbb{1}(\mathrm{P}_{A|XC}\cdot \mathrm{P}_{C|Z}) \\
		&= \mathrm{E}_{K|C}\cdot \mathrm{P}_{C|Z} \;.
	\end{split}	
	\]
	$\mathcal{F}\otimes \mathbb{1}(\mathrm{P}_{AC|XZ})$ is defined in a similar way.
\end{defn}

More explicitly, the diamond norm can be written in the following way.
\begin{widetext}
\begin{equation}\label{eq:diamond-norm-explicit}
	\begin{split}
		\| \mathcal{E}-\mathcal{F}\|_{\diamond} & =\underset{\mathrm{P}_{AC|XZ}}{\mathrm{max}}\| \left(\mathcal{E}-\mathcal{F}\right)\otimes \mathbb{1}(\mathrm{P}_{AC|XZ} )\|_1\\
		&= \underset{\mathrm{P}_{AC|XZ}}{\mathrm{max}}\| \mathrm{E}_{K|C}\cdot \mathrm{P}_{C|Z} - \mathrm{F}_{K|C}\cdot \mathrm{P}_{C|Z}\|_1\\
		&=\underset{\mathrm{P}_{AC|XZ}}{\mathrm{max}} \frac{1}{2} \sum_k \underset{z}{\mathrm{max}} \sum_c \mathrm{P}_{C|Z}(c|z) \left| \mathrm{E}_{K|C}(k|c)-\mathrm{F}_{K|C}(k|c) \right| \\
		&= \underset{\mathrm{P}_{AC|XZ}}{\mathrm{max}} \frac{1}{2} \sum_k \underset{z}{\mathrm{max}} \sum_c \mathrm{P}_{C|Z}(c|z) \times \\
		& \qquad \qquad \times \left| \sum_x \mathrm{Pr}_{\mathcal{E}}(x) \sum_{a | \mathcal{E}(a,x)=k} \mathrm{P}_{A|XC}(a|xc)  -\sum_x \mathrm{Pr}_{\mathcal{F}}(x) \sum_{a | \mathcal{F}(a,x)=k} \mathrm{P}_{A|XC}(a|xc)  \right|
	\end{split}
\end{equation}
\end{widetext}
where the third equality is due to the explicit form of the trace distance previously given in \cite{masanes2009universally,hanggi2009quantum}. 

As in Definition \ref{def:permutation-invariant-test}, we say that a channel is permutation invariant if for all permutations $\pi$, $\mathcal{E}(\mathrm{P}_{A|X})=\mathcal{E}(\mathrm{P}_{A|X}\circ\pi)$. In a similar manner we can also consider channels which are $\mathcal{S}$ invariant. 
\begin{defn}\label{def:s-invariant}
	We say that a mapping $\mu$ of $(a,x)$ to $(a',x')$ respects the symmetry $\mathcal{S}$ if for every state $\mathrm{P}_{A|X}$ \emph{with this symmetry} $\mathrm{P}_{A|X}=\mathrm{P}_{A|X}\circ\mu$. 
	
	A channel  $\mathcal{E}$ is $\mathcal{S}$ invariant if for \emph{every} $\mathrm{P}_{A|X}$ and every mapping $\mu$ which respects $\mathcal{S}$ we have $\mathcal{E}(\mathrm{P}_{A|X})=\mathcal{E}(\mathrm{P}_{A|X}\circ\mu)$.
\end{defn}
For example, when considering the CHSH symmetry in Definition \ref{def:chsh-symmetry} one such possible mapping $\mu$ may map $(a_i,b_i,x_i,y_i)=(0,0,0,0)$ to $(a_i,b_i,x_i,y_i)=(1,0,1,1)$ for every $i\in[n]$ (see Figure \ref{fig:CHSH_symmetry}). 

Using these concepts and the de Finetti reduction given in Theorem \ref{thm:post-selection} we can prove the following bound on the diamond norm. 
\begin{thm}\label{thm:diamond-norm-thm}
	For any two permutation invariant and $\mathcal{S}$ invariant channels $\mathcal{E},\mathcal{F}:\{\mathrm{P}_{A|X}\}\rightarrow\{\mathrm{P}_K\}$
	\[
		\| \mathcal{E}-\mathcal{F}\|_{\diamond} \leq (n+1)^d \underset{\tau^{\mathcal{S}}_{AC|XZ}}{\mathrm{max}}\| \left(\mathcal{E}-\mathcal{F}\right)\otimes \mathbb{1}(\tau^{\mathcal{S}}_{AC|XZ} )\|_1
	\]
	where $d$ is the number of degrees of freedom of $\mathcal{S}$ and $\tau^{\mathcal{S}}_{AC|XZ}$ is a non-signalling extension of the de Finetti state $\tau^{\mathcal{S}}_{A|X}$. 
\end{thm}

In order to prove Theorem \ref{thm:diamond-norm-thm} we first prove the following lemma. 
\begin{lem}\label{lem:trace-distance-lem}
	For every two permutation invariant and $\mathcal{S}$ invariant (as in Definition \ref{def:s-invariant}) channels $\mathcal{E},\mathcal{F}:\{\mathrm{P}_{A|X}\}\rightarrow\{\mathrm{P}_K\}$ where $\mathrm{P}_K$ is a probability distribution over $k\in\{0,1\}^t$ for some $t>0$, $\forall \mathrm{P}_{AC|XZ}$,
	\[
		 \| \left(\mathcal{E}-\mathcal{F}\right)\otimes \mathbb{1}(\mathrm{P}_{AC|XZ} )\|_1 \leq (n+1)^d \| \left(\mathcal{E}-\mathcal{F}\right)\otimes \mathbb{1}(\tau^{\mathrm{P}_{AC|XZ}}_{AC|XZ} )\|_1
	\]
	where $\tau^{\mathrm{P}_{AC|XZ}}_{AC|XZ}$ is a non-signalling extension of $\tau^{\mathcal{S}}_{A|X}$ which depends on the specific state $\mathrm{P}_{AC|XZ}$.
\end{lem}

\begin{proof}
	First, as in the previous proof and Equation \eqref{eq:only-permutations}, since the channels are permutation invariant and $\mathcal{S}$ invariant it is sufficient to consider states $\mathrm{P}_{A|X}$ which are permutation invariant and have the symmetry $\mathcal{S}$. 
	
	Given a specific state $\mathrm{P}_{AC|XZ}$, according to Lemma \ref{lem:extension}, we can see this extension as a set of convex decompositions of $\mathrm{P}_{A|X}$. That is, every possible input $z$ indicates a specific decomposition $\{ ( p_{c_z}, \mathrm{P}^{c_z}_{A|X}) \}_{c_z}$ such that $p_{c_z}=\mathrm{P}_{C|Z}(c_z|z)$ and $\mathrm{P}^{c_z}_{A|X}(a|x)=\mathrm{P}_{AC|XZ}(a,c_z|x,z)$. Since this is a convex decomposition of $\mathrm{P}_{A|X}$ we also have 
	\begin{equation} \label{eq:decomposition}
		\forall z \quad \sum_c p_c \cdot \mathrm{P}^{c}_{A|X} = \mathrm{P}_{A|X} \;.
	\end{equation}
	We now use the set of decompositions of $\mathrm{P}_{A|X}$ to construct a set of decompositions of the de Finetti state $\tau^{\mathcal{S}}_{A|X}$. Combining Theorem \ref{thm:post-selection} with Lemma \ref{lem:extension_condition} and Lemma \ref{lem:extension} we know that there exists a non-signalling state $\mathrm{R}_{A|X}$ such that 
	\[
	\begin{split}
		\tau^{\mathcal{S}}_{A|X} &= \frac{1}{(n+1)^d}\mathrm{P}_{A|X} + \left(1- \frac{1}{(n+1)^d} \right)\mathrm{R}_{A|X} \\
		&=\frac{1}{(n+1)^d}\sum_c p_c \cdot \mathrm{P}^{c}_{A|X} + \left(1- \frac{1}{(n+1)^d} \right)\mathrm{R}_{A|X} \;.
	\end{split}
	\]
	where the second equality is due to Equation \eqref{eq:decomposition}. For every $z$ this defines a decomposition $\{ ( \frac{1}{(n+1)^d} \cdot p_{c_z}, \mathrm{P}^{c_z}_{A|X} ) \}_{c_z} \cup \{ (1- \frac{1}{(n+1)^d}, \mathrm{R}_{A|X}) \} $ of $\tau^{\mathcal{S}}_{A|X}$. That is, this defines an extension $\tau^{\mathrm{P}_{AC|XZ}}_{AC'|XZ}$ of  $\tau^{\mathcal{S}}_{A|X}$ where $C'=C\cup \{c'\}$. 
	
	This connection between the extensions $\mathrm{P}_{AC|XZ}$ and $\tau^{\mathrm{P}_{AC|XZ}}_{AC'|XZ}$ allow us to get the bound on the trace distance and prove the lemma:
	\begin{widetext}
	
	\[
	\begin{split}
		\| \left(\mathcal{E}-\mathcal{F}\right)\otimes \mathbb{1}(\tau^{\mathrm{P}_{AC|XZ}}_{AC'|XZ} )\|_1 &=\frac{1}{2} \sum_k \underset{z}{\mathrm{max}} \sum_{c\in C'} \tau^{\mathrm{P}_{AC|XZ}}_{C'|Z}(c|z) \times  \\
		&  \qquad \times \left| \sum_x \mathrm{Pr}_{\mathcal{E}}(x) \sum_{a | \mathcal{E}(a,x)=k} \tau^{\mathrm{P}_{AC|XZ}}_{A|XC'}(a|xc)  -\sum_x \mathrm{Pr}_{\mathcal{F}}(x) \sum_{a | \mathcal{F}(a,x)=k} \tau^{\mathrm{P}_{AC|XZ}}_{A|XC'}(a|xc)  \right|  \\
		&=\frac{1}{2} \sum_k \underset{z}{\mathrm{max}} \left[ \sum_{c\in C} \tau^{\mathrm{P}_{AC|XZ}}_{C'|Z}(c|z) \times \right.  \\
		&  \qquad \times \left| \sum_x \mathrm{Pr}_{\mathcal{E}}(x) \sum_{a | \mathcal{E}(a,x)=k} \tau^{\mathrm{P}_{AC|XZ}}_{A|XC'}(a|xc)  -\sum_x \mathrm{Pr}_{\mathcal{F}}(x) \sum_{a | \mathcal{F}(a,x)=k} \tau^{\mathrm{P}_{AC|XZ}}_{A|XC'}(a|xc)  \right|  + \\
		& \qquad + \left(1- \frac{1}{(n+1)^d} \right) \times \\
		&  \left. \qquad \times \left| \sum_x \mathrm{Pr}_{\mathcal{E}}(x) \sum_{a | \mathcal{E}(a,x)=k} \tau^{\mathrm{P}_{AC|XZ}}_{A|XC'}(a|xc')  -\sum_x \mathrm{Pr}_{\mathcal{F}}(x) \sum_{a | \mathcal{F}(a,x)=k} \tau^{\mathrm{P}_{AC|XZ}}_{A|XC'}(a|xc')  \right| \right]  \\
		&\geq \frac{1}{2} \sum_k \underset{z}{\mathrm{max}}  \sum_{c\in C} \tau^{\mathrm{P}_{AC|XZ}}_{C'|Z}(c|z) \times   \\
		&  \qquad \times \left| \sum_x \mathrm{Pr}_{\mathcal{E}}(x) \sum_{a | \mathcal{E}(a,x)=k} \tau^{\mathrm{P}_{AC|XZ}}_{A|XC'}(a|xc)  -\sum_x \mathrm{Pr}_{\mathcal{F}}(x) \sum_{a | \mathcal{F}(a,x)=k} \tau^{\mathrm{P}_{AC|XZ}}_{A|XC'}(a|xc)  \right| 
	\end{split}
	\]
	\begin{equation}\label{eq:trace-distance-formula}
	\begin{split}
		&=\frac{1}{2} \sum_k \underset{z}{\mathrm{max}}  \sum_{c\in C} \frac{1}{(n+1)^d}\cdot \mathrm{P}_{C|Z}(c|z) \times   \\
		&  \qquad \times \left| \sum_x \mathrm{Pr}_{\mathcal{E}}(x) \sum_{a | \mathcal{E}(a,x)=k} \mathrm{P}_{A|XC}(a|xc)  -\sum_x \mathrm{Pr}_{\mathcal{F}}(x) \sum_{a | \mathcal{F}(a,x)=k} \mathrm{P}_{A|XC}(a|xc)  \right|  \\
		&= \frac{1}{(n+1)^d} \| \left(\mathcal{E}-\mathcal{F}\right)\otimes \mathbb{1}(\mathrm{P}_{AC|XZ} )\|_1 \;.
	\end{split}
	\end{equation}
	where in order to get the second equality we divide the sum over $C'=C\cup \{c'\}$ to the sum over $C$ and then additional part of the partition $c'$. The next inequality is then correct since
	\[
		 \left(1- \frac{1}{(n+1)^d} \right) \left| \sum_x \mathrm{Pr}_{\mathcal{E}}(x) \sum_{a | \mathcal{E}(a,x)=k} \tau^{\mathrm{P}_{AC|XZ}}_{A|XC}(a|xc')  -\sum_x \mathrm{Pr}_{\mathcal{F}}(x) \sum_{a | \mathcal{F}(a,x)=k} \tau^{\mathrm{P}_{AC|XZ}}_{A|XC}(a|xc')  \right| \geq 0
	\]
	 and the two last equalities are due to the specific decomposition of $\tau^{\mathcal{S}}_{A|X}$ that we defined and the definition of the trace distance. The lemma then follows from Equation \eqref{eq:trace-distance-formula}. \qedhere
\end{widetext}
\end{proof}

Theorem \ref{thm:diamond-norm-thm} now easily follows from Lemma \ref{lem:trace-distance-lem}:
\begin{proof}[Proof of Theorem \ref{thm:diamond-norm-thm}]
	Using Lemma \ref{lem:trace-distance-lem}, 
	\[
	\begin{split}
		\| \mathcal{E} - \mathcal{F} \|_{\diamond} &= \underset{\mathrm{P}_{AC|XZ}}{\mathrm{max}}\| \left(\mathcal{E}-\mathcal{F}\right)\otimes \mathbb{1}(\mathrm{P}_{AC|XZ} )\|_1\\
		&\leq (n+1)^d \underset{\tau^{\mathrm{P}_{AC|XZ}}_{AC'|XZ}}{\mathrm{max}}\| \left(\mathcal{E}-\mathcal{F}\right)\otimes \mathbb{1}(\tau^{\mathrm{P}_{AC|XZ}}_{AC'|XZ})\|_1\\
		&\leq (n+1)^d \underset{\tau_{AC|XZ}}{\mathrm{max}}\| \left(\mathcal{E}-\mathcal{F}\right)\otimes \mathbb{1}(\tau_{AC|XZ})\|_1\\
	\end{split}
	\]
	where $\tau_{AC|XZ}$ is a non-signalling extension of $\tau^{\mathcal{S}}_{A|X}$. 
\end{proof}

Theorem \ref{thm:diamond-norm-thm} can be applied to simplify the analysis of device independent cryptography. In particular, when considering protocols which are based on the CHSH symmetry we can use the following Corollary.

\begin{cor}\label{cor:chsh-diamond-bound}
	For any two permutation invariant and CHSH invariant channels $\mathcal{E},\mathcal{F}:\{\mathrm{P}_{AB|XY}\}\rightarrow\{\mathrm{P}_K\}$
	\[
		\| \mathcal{E}-\mathcal{F}\|_{\diamond} \leq (n+1) \underset{\tau^{\mathcal{CHSH}}_{ABC|XYZ}}{\mathrm{max}}\| \left(\mathcal{E}-\mathcal{F}\right)\otimes \mathbb{1}(\tau^{\mathcal{CHSH}}_{ABC|XYZ} )\|_1
	\]
	where $\tau^{\mathcal{CHSH}}_{ABC|XYZ}$ is a non-signalling extension of the de Finetti state $\tau^{\mathcal{CHSH}}_{AB|XY}$. 
\end{cor}

Corollary \ref{cor:chsh-diamond-bound} implies that when proving security of cryptographic protocols based on the CHSH inequality it is sufficient to consider the case where Alice and Bob share the de Finetti state $\tau^{\mathcal{CHSH}}_{AB|XY}$. However, one still needs to take into account all possible non-signalling extensions of this bipartite state to a tripartite state $\tau^{\mathcal{CHSH}}_{ABC|XYZ}$ that includes the adversary, as can be seen from the maximisation over $\tau^{\mathcal{CHSH}}_{ABC|XYZ}$ . These type of proofs can be done, as for example in \cite{hanggi2009quantum}.

We further emphasise that this does not imply that Alice and Bob's state is a convex combination of i.i.d.\ states when including the adversary's knowledge, but only from Alice and Bob's point of view. This is in contrast to the stronger result achieved by the quantum post selection theorem \cite{christandl2009postselection}. However, due to the no-go theorems given in  \cite{hanggi2010impossibility,arnon2012limits} we know that such a stronger result is not possible in the more general scenario that we consider here.

Two additional remarks are in order. First, to use this corollary we must consider protocols which are invariant under the CHSH symmetry (and therefore the channel describing them will also be invariant under the relevant mappings). Fortunately, this invariance can be ensured by performing an additional step in the beginning of the protocol, called depolarisation \cite{Masanes2008depolarization}. The depolarisation procedure will not affect the correctness of the protocol and will make it invariant under the appropriate mappings $\mu$. Such depolarisation procedures can also be constructed for other types of protocols such as protocols which are based on the chained Bell inequalities.  

Second, note that since the state $\tau^{\mathcal{CHSH}}_{AB|XY}$ is not quantum but non-signalling, this result cannot be applied in a trivial manner to proofs where it is assumed that Alice and Bob's statistics is restricted by quantum theory.

\bibliography{refs.bib}

\end{document}